%% file: ms.tex
\documentclass[AMA, STIX1COL]{WileyNJD-v2}

\usepackage{float}
\usepackage{amsthm}

\theoremstyle{definition}
\newtheorem{example}{Example}

\newcommand{\proofs}[1][A]{#1ppendix}
\newcommand{\REACTlink}{\url{https://github.com/Monoxido45/REACT}}

\def\ourapproach{{\texttt{REACT}}}
\def\fourapproach{{\texttt{fREACT}}}
\def\bourapproach{{\texttt{bREACT}}}

\def\half{\frac{1}{2}}

\def\I{{\mathbb I}}

\def\D{{\mathcal{D}}}

\def\P{{\mathbb P}}

\def\Re{{\mathbb R}}

\def\t0{{\theta_0}}

\def\Tz{{\Theta_0}}
\def\Tzc{{\Theta_0^c}}
\def\limn{\lim_{n \rightarrow \infty}}

\usepackage{tikz}

\newcommand{\reactSimple}{
\begin{tikzpicture}[x=0.75pt,y=0.75pt,yscale=-1,xscale=1, every picture/.style={line width=0.75pt}]

\draw  [color={rgb, 255:red, 74; green, 144; blue, 226 }  ,draw opacity=0 ][fill={rgb, 255:red, 74; green, 144; blue, 226 }  ,fill opacity=0.2 ] (69.67,-0.42) -- (270.33,-0.42) -- (270.33,308.58) -- (69.67,308.58) -- cycle ;
\draw [line width=1.5]    (11.33,70.5) -- (427.33,70.5) ;
\draw [shift={(431.33,70.5)}, rotate = 180] [fill={rgb, 255:red, 0; green, 0; blue, 0 }  ][line width=0.08]  [draw opacity=0] (13.4,-6.43) -- (0,0) -- (13.4,6.44) -- (8.9,0) -- cycle    ;
\draw [line width=1.5]    (11.33,160.5) -- (427.33,160.5) ;
\draw [shift={(431.33,160.5)}, rotate = 180] [fill={rgb, 255:red, 0; green, 0; blue, 0 }  ][line width=0.08]  [draw opacity=0] (13.4,-6.43) -- (0,0) -- (13.4,6.44) -- (8.9,0) -- cycle    ;
\draw [line width=1.5]    (11.33,250.17) -- (427.33,250.17) ;
\draw [shift={(431.33,250.17)}, rotate = 180] [fill={rgb, 255:red, 0; green, 0; blue, 0 }  ][line width=0.08]  [draw opacity=0] (13.4,-6.43) -- (0,0) -- (13.4,6.44) -- (8.9,0) -- cycle    ;
\draw    (466.67,10) -- (466.33,300.33) ;
\draw    (169.67,60.33) -- (169.67,80.33) ;
\draw    (169.67,150.33) -- (169.67,170.33) ;
\draw    (269.67,60.33) -- (269.67,80.33) ;
\draw    (70.67,60.33) -- (70.67,80.33) ;
\draw    (70.67,150.33) -- (70.67,170.33) ;
\draw    (269.67,150.33) -- (269.67,170.33) ;
\draw    (169.67,240.67) -- (169.67,260.67) ;
\draw    (70.67,240.67) -- (70.67,260.67) ;
\draw    (269.67,240.67) -- (269.67,260.67) ;
\draw [color={rgb, 255:red, 65; green, 117; blue, 5 }  ,draw opacity=1 ][line width=2.25]    (81.5,70.75) -- (201,70.75) ;
\draw [shift={(201,70.75)}, rotate = 180] [color={rgb, 255:red, 65; green, 117; blue, 5 }  ,draw opacity=1 ][line width=2.25]      (7.83,-7.83) .. controls (3.51,-7.83) and (0,-4.32) .. (0,0) .. controls (0,4.32) and (3.51,7.83) .. (7.83,7.83) ;
\draw [shift={(81.5,70.75)}, rotate = 0] [color={rgb, 255:red, 65; green, 117; blue, 5 }  ,draw opacity=1 ][line width=2.25]      (7.83,-7.83) .. controls (3.51,-7.83) and (0,-4.32) .. (0,0) .. controls (0,4.32) and (3.51,7.83) .. (7.83,7.83) ;
\draw [color={rgb, 255:red, 149; green, 0; blue, 18 }  ,draw opacity=1 ][line width=2.25]    (291.5,160.25) -- (411,160.25) ;
\draw [shift={(411,160.25)}, rotate = 180] [color={rgb, 255:red, 149; green, 0; blue, 18 }  ,draw opacity=1 ][line width=2.25]      (7.83,-7.83) .. controls (3.51,-7.83) and (0,-4.32) .. (0,0) .. controls (0,4.32) and (3.51,7.83) .. (7.83,7.83) ;
\draw [shift={(291.5,160.25)}, rotate = 0] [color={rgb, 255:red, 149; green, 0; blue, 18 }  ,draw opacity=1 ][line width=2.25]      (7.83,-7.83) .. controls (3.51,-7.83) and (0,-4.32) .. (0,0) .. controls (0,4.32) and (3.51,7.83) .. (7.83,7.83) ;
\draw [color={rgb, 255:red, 215; green, 197; blue, 0 }  ,draw opacity=1 ][line width=2.25]    (193.5,250.25) -- (313,250.25) ;
\draw [shift={(313,250.25)}, rotate = 180] [color={rgb, 255:red, 215; green, 197; blue, 0 }  ,draw opacity=1 ][line width=2.25]      (7.83,-7.83) .. controls (3.51,-7.83) and (0,-4.32) .. (0,0) .. controls (0,4.32) and (3.51,7.83) .. (7.83,7.83) ;
\draw [shift={(193.5,250.25)}, rotate = 0] [color={rgb, 255:red, 215; green, 197; blue, 0 }  ,draw opacity=1 ][line width=2.25]      (7.83,-7.83) .. controls (3.51,-7.83) and (0,-4.32) .. (0,0) .. controls (0,4.32) and (3.51,7.83) .. (7.83,7.83) ;

\draw (431,144.07) node [anchor=north west][inner sep=0.75pt]  [font=\Large]  {$\boldsymbol{\phi }$};
\draw (431.67,54.73) node [anchor=north west][inner sep=0.75pt]  [font=\Large]  {$\boldsymbol{\phi }$};
\draw (431,234.07) node [anchor=north west][inner sep=0.75pt]  [font=\Large]  {$\boldsymbol{\phi }$};
\draw (510,3) node [anchor=north west][inner sep=0.75pt]   [align=left] {\textbf{{\LARGE \underline{Decision}}}};
\draw (481.67,60) node [anchor=north west][inner sep=0.75pt]  [font=\LARGE] [align=left] {Accept $\displaystyle H_{0}$};
\draw (481.67,150) node [anchor=north west][inner sep=0.75pt]  [font=\LARGE] [align=left] {Reject $\displaystyle H_{0}$};
\draw (481.67,240) node [anchor=north west][inner sep=0.75pt]  [font=\LARGE] [align=left] {Remain Agnostic};
\draw (163.5,80) node [anchor=north west][inner sep=0.75pt]   [align=left] {{\large \textbf{{\fontfamily{pcr}\selectfont 0}}}};
\draw (263,80.07) node [anchor=north west][inner sep=0.75pt]    {$\large\boldsymbol{\Delta }$};
\draw (52.33,80.07) node [anchor=north west][inner sep=0.75pt]    {$\large\boldsymbol{-\Delta }$};
\draw (163.5,170) node [anchor=north west][inner sep=0.75pt]   [align=left] {{\large \textbf{{\fontfamily{pcr}\selectfont 0}}}};
\draw (263,170.4) node [anchor=north west][inner sep=0.75pt]    {$\large\boldsymbol{\Delta }$};
\draw (52.33,170.4) node [anchor=north west][inner sep=0.75pt]    {$\large\boldsymbol{-\Delta }$};
\draw (163.5,261) node [anchor=north west][inner sep=0.75pt]   [align=left] {{\large \textbf{{\fontfamily{pcr}\selectfont 0}}}};
\draw (263,260.73) node [anchor=north west][inner sep=0.75pt]    {$\large\boldsymbol{\Delta }$};
\draw (52.33,260.73) node [anchor=north west][inner sep=0.75pt]    {$\large\boldsymbol{-\Delta }$};
\draw (181.5,61.5) node  [color={rgb, 255:red, 65; green, 117; blue, 5 }  ,opacity=1 ] [align=left] {\begin{minipage}[lt]{68pt}\setlength\topsep{0pt}
{\LARGE \textbf{\textit{{\fontfamily{pcr}\selectfont C}}}}
\end{minipage}};
\draw (389.5,150.5) node  [color={rgb, 255:red, 149; green, 0; blue, 18 }  ,opacity=1 ] [align=left] {\begin{minipage}[lt]{68pt}\setlength\topsep{0pt}
{\fontfamily{pcr}\selectfont {\LARGE \textbf{\textit{C}}}}
\end{minipage}};
\draw (293.5,240.5) node  [color={rgb, 255:red, 215; green, 197; blue, 0 }  ,opacity=1 ] [align=left] {\begin{minipage}[lt]{68pt}\setlength\topsep{0pt}
{\fontfamily{pcr}\selectfont {\LARGE \textbf{\textit{C}}}}
\end{minipage}};

\end{tikzpicture}
}

\usepackage{blindtext}

\usepackage{float}
\usepackage{graphicx}
\graphicspath{{figures/}{../figures/}}
\usepackage{caption}
\usepackage{subcaption}

\DeclareFontFamily{U}{mathb}{\hyphenchar\font45}
\DeclareFontShape{U}{mathb}{m}{n}{ <-6> matha5 <6-7> matha6 <7-8>
mathb7 <8-9> mathb8 <9-10> mathb9 <10-12> mathb10 <12-> mathb12 }{}
\DeclareSymbolFont{mathb}{U}{mathb}{m}{n}

\DeclareMathAccent{\abxring}{0}{mathb}{"38}

\DeclareFontFamily{U}{mathb}{\hyphenchar\font45}
\DeclareFontShape{U}{mathb}{m}{n}{ <-6> matha5 <6-7> matha6 <7-8>
mathb7 <8-9> mathb8 <9-10> mathb9 <10-12> mathb10 <12-> mathb12 }{}
\DeclareSymbolFont{mathb}{U}{mathb}{m}{n}

\begin{document}

\input{sections/title}

\maketitle


\input{sections/main_text}

\subsection*{ACKNOWLEDGEMENTS}

We thank Luís Gustavo Esteves, Julio Michael Stern and Victor Coscrato for the fruitful discussion about \ourapproach.

\subsection*{Author contributions}

Rafael Izbicki was responsible for conceiving this paper, developing the methods and theory, and writing the manuscript.
Luben M. C. Cabezas was responsible for developing the methods and theory, developing computer code, analyzing and interpreting the results of both applications, and reviewing the manuscript.
Fernando Colugnati was responsible for the meta-analysis case study data and interpretation, writing the manuscript, and review.
Rodrigo F. L. Lassance was responsible for developing the methods and theory, developing computer code, analyzing and interpreting the results of both applications, and writing the manuscript.
Altay L. de Souza was responsible for collaborating on the study cases, writing, and reviewing the manuscript.
Rafael Stern was responsible for conceiving this paper, developing the methods and theory, and writing the manuscript. 

\subsection*{Financial disclosure}

None reported.

\subsection*{Conflict of interest}

The authors declare no potential conflict of interests.





\appendix

\input{sections/appendix}

\input{sections/proofs}

\bibliography{wileyNJD-AMA}%

\clearpage

\end{document}

%% file: sections/title.tex
\title{\ourapproach{} to NHST: Sensible conclusions for meaningful hypotheses}

\author[1]{Rafael Izbicki*}

\author[1,2]{Luben M. C. Cabezas}

\author[3]{Fernando A. B. Colugnatti}

\author[1,2]{Rodrigo F. L. Lassance}
\author[4]{Altay A. L. de Souza}
\author[5]{Rafael B. Stern}

\authormark{RAFAEL IZBICKI \textsc{et al}}

\address[1]{\orgdiv{Department of Statistics}, \orgname{Federal University of São Carlos}, \orgaddress{\state{São Paulo}, \country{Brazil}}}
\address[2]{\orgdiv{Institute of Mathematics and Computer Sciences}, \orgname{University of São Paulo}, \orgaddress{\state{São Paulo}, \country{Brazil}}}
\address[3]{\orgdiv{School of Medicine}, \orgname{Federal University of Juiz de Fora}, \orgaddress{\state{Minas Gerais}, \country{Brazil}}}
\address[4]{\orgdiv{Psychobiology Department}, \orgname{Federal University of São Paulo}, \orgaddress{\state{São Paulo}, \country{Brazil}}}
\address[5]{\orgdiv{Institute of Mathematics and Statistics}, \orgname{University of São Paulo}, \orgaddress{\state{São Paulo}, \country{Brazil}}}

\corres{*Rafael Izbicki  \email{rafaelizbicki@gmail.com}}

\fundingInfo{FAPESP (grants 2019/11321-9, 2022/08579-7 and 2023/07068-1) and CNPq (grants 309607/2020-5 and 422705/2021-7). This study was financed in part by the Coordenação de Aperfeiçoamento de Pessoal de Nível Superior - Brasil (CAPES) - Finance Code 001.
}
\abstract[Summary]{
\input{sections/abstract}
}

\keywords{hypothesis tests, NHST, p-values, equivalence tests, three-way decision procedures}

%% file: sections/abstract.tex




While Null Hypothesis Significance Testing (NHST) remains a widely used statistical tool, it suffers from several shortcomings in its common usage, such as conflating statistical and practical significance, the formulation of inappropriate null hypotheses, and the inability to distinguish between accepting the null hypothesis and failing to reject it. Recent efforts have focused on developing alternatives that address these issues. Despite these efforts, conventional NHST remains dominant in scientific research due to its procedural simplicity and mistakenly presumed ease of interpretation.
Our work presents an intuitive alternative to conventional NHST designed to bridge the gap between the expectations of researchers and the actual outcomes of hypothesis tests: \ourapproach{}. \ourapproach\ not only tackles  shortcomings of conventional NHST but also offers additional advantages over existing alternatives. For instance, \ourapproach\ applies to multiparametric hypotheses and does not require stringent significance-level corrections when conducting multiple tests. We illustrate the practical utility of \ourapproach\  through real-world data examples.

%% file: sections/main_text.tex
\section{Introduction}\label{sec:intro}

Statistical hypothesis testing is a fundamental tool in scientific research, offering a structured approach to tackling research questions. In the field of clinical research, it goes beyond being a mere recommendation and becomes a nearly-mandatory requirement for publishing results and aiding important decision-making. For example, for many years, clinical trial design has remained heavily reliant on calculating sample sizes by evaluating the statistical power of hypothesis tests.

However, there has been a growing wave of criticism directed at conventional Null Hypothesis Significance Testing\footnote{By ``conventional NHST" we refer to the approach typically used by practitioners of NHST, which involves (i) the formulation of a null hypothesis usually pointing to no effects (e.g., $H_0: \mu_A = \mu_B$ in a two-sample tests problem), and (ii) the establishment of a statistical test with two outcomes for $H_0$ that controls the type I error at a predefined level of $\alpha$.} (NHST) and p-values \citep{hays1963statistics,Wasserman2013,Trafimow2018,Pike2019,Greenland2016,Kadane2016,Wasserstein2019, Diggle2011}. Much of this criticism arises from the misuse and misinterpretation of statistical tests. Even in meta-analytic studies, hypothesis tests often lead to misinterpretations that have a high impact on public policymaking. As a result, many scientific journals have taken a position against the use of conventional NHST, often discouraging its use \citep{campo2008interpretation,Trafimow2015}.

Despite the ongoing criticisms and the several alternatives proposed (see Section \ref{sec:relatedWork} for a review), conventional NHST remains the widely accepted standard in scientific research. 
This continued popularity can be explained by various factors, including its operational simplicity. Additionally, though NHST results are nuanced and complex, they often give a false impression of being easy to interpret.

It is therefore crucial to provide an alternative to conventional NHST while keeping as much of its operational simplicity as possible.
This alternative should bridge the gap between researchers' expectations and the actual outcomes of tests, preventing misinterpretations. For instance, distinguishing between accepting and failing to reject the null hypothesis is essential in practice. Therefore, the main goal of this work is to introduce a framework for hypothesis testing, \ourapproach, that better meets the needs of researchers. \ourapproach\ builds upon existing solutions that aim to improve NHST, incorporating elements from three-way tests and equivalence tests like Two One-Sided Tests (TOST; \citealt{schuirmann1987comparison}). For a detailed comparison with these and other methods,  see Section \ref{sec:relatedWork}. We also provide an R package that implements \ourapproach\ for common models.

Next, we revisit some of the major concerns of the standard approach to hypothesis testing. Section \ref{sec:react_simple} then introduces a simplified version of \ourapproach\ for a single hypothesis that only concerns one parameter. Section \ref{sec:relatedWork} shows how \ourapproach\ relates to other methods in the literature. Section \ref{sec:methods} introduces our full procedure, and presents its properties. Two applications are presented in Section \ref{sec:applications}. Section \ref{sec:final} concludes the paper.

\subsection{Review of some NHST issues}

\paragraph{Issue 1: Statistical significance versus practical significance.} 
One of the primary challenges of NHST is the difficulty in distinguishing between statistical significance and practical significance \citep{wasserstein2016asa}. 
This has been noted very early in psychology \citep{hays1963statistics}. 
Still, low p-values are often
used as a proxy for important practical significance.
As an example, a study on the effectiveness of aspirin in preventing myocardial infarction \citep{bartolucci2011meta,sullivan2012using}  found statistically significant results (p-value $<$ .00001), and therefore was stopped early due to conclusive evidence. As a result, many people were advised to take aspirin to prevent heart attacks. However, upon further investigation, the effect size was found to be practically insignificant, and the recommendation had to be revised \citep{sullivan2012using}. This raises important concerns about the practical relevance of the output of such hypothesis tests  and their implications for public health.

\paragraph{Issue 2: Implausibility of the null hypothesis.} 
There are very few situations in which one expects precise null hypotheses to be exactly true \citep{Edwards1963,amrhein2017earth,lecoutre2022significance}. 
For example, when comparing two medications, the primary concern is typically whether they are practically equivalent, as it is highly unlikely that any two medications will produce precisely the same effects. Therefore, establishing whether the medications have similar outcomes is often more relevant than attempting to evaluate if they have the same effect on average. Indeed, bioequivalence tests effectively modify the null hypothesis to align with practical equivalence (refer to Section \ref{sec:relatedWork}). This is also true in other domains. Indeed,
\citet{cohen1992things} mentions that
\begin{quote}
    ``The null hypothesis, taken literally (and that's the only way you can take it in formal hypothesis testing), is always false in the real world. It can only be true in the bowels of a computer processor running a Monte Carlo study (and even then a stray electron may make it false). If it is false, even to a tiny degree, it must be the case that a large enough sample will produce a significant result and lead to its rejection. So if the null is always false, what's the big deal about rejecting it?"
\end{quote}

A consequence of this fact is that, provided that the sample size is large enough, in most problems one will always reject the null hypothesis, making the results of NHST not meaningful \citep{vaughan1969beyond,cohen1992things,cohen1994earth,gill1999insignificance,Faber2014}.

\paragraph{Issue 3: Accept versus not-reject the null hypothesis.} 
The absence of evidence is not evidence of absence \citep{altman1995statistics}, but NHST is not able to differentiate between failure to reject 
 the null hypothesis and its acceptance. Indeed, \citet{Edwards1963} point out that
\begin{quote}
    ``If the null hypothesis is not rejected, it
remains in a kind of limbo of suspended disbelief."
\end{quote}
This is because a null may not be rejected either because the test is not powerful enough to reject it  (e.g. due to a small sample size) or because $H_0$ is indeed true. 
This has major implications for public policymaking, especially during times of crisis, such as the COVID-19 pandemic. It can be challenging to interpret research findings on the efficacy of interventions, such as the use of masks \citep{jefferson2023physical} or hydroxychloroquine \citep{mehra2020retracted}, when there is a lack of consensus on what constitutes evidence of absence versus no evidence of an effect \citep{fidler2018epistemic}. 
 It is therefore very important for researchers to distinguish between ``no evidence of effect" and ``evidence of no effect" when interpreting research results \citep{altman1995statistics,keysers2020using}.
Plain NHST cannot do that on its own.
 
\vspace{2mm}

In this paper, we overcome these issues by introducing the Region of Equivalence Agnostic Confidence-based Test (\ourapproach).

\subsection{Novelty}

Numerous solutions have been proposed to tackle issues  1, 2, and 3; see Section \ref{sec:relatedWork} for a detailed overview of how such approaches relate to our approach. 
In particular, \ourapproach\ uses the strengths of equivalence and three-way hypothesis testing.  
 However, to the best of our knowledge, \ourapproach\ is the first approach that simultaneously solves these and other issues.
Specifically, \ourapproach:
\begin{itemize}
\item Is designed to work with meaningful null hypotheses that encode practical significance
\item Clearly distinguishes between ``evidence of absence'' and ``absence of evidence''
\item Does not need ad hoc procedures to perform multiple comparisons. In fact, \ourapproach\ not only automatically controls the Family-Wise Error Rate (FWER) of false rejections at $\alpha$ (type I errors), but it also
controls the FWER of false acceptances at $\alpha$ (type II errors)
\item Can be easily applied to hypotheses that involve several parameters, such as in an ANOVA setting
\item Leads to fully logically coherent solutions. For instance, in a multiple comparison problem, if \ourapproach\ rejects the null $\mu_1=\mu_2$, it will also reject the null $\mu_1=\mu_2=\mu_3$.
This level of coherence is not typically achieved with standard procedures, resulting in epistemic confusion that complicates the reporting of test results.
\end{itemize}

\ourapproach\  requires the specification of two components: 
 (i) a confidence region for the parameters of interest, and (ii) a null hypothesis that reflects a range of
parameter values considered to be practically equivalent (e.g., in an ANOVA context, $H_0: |\mu_A-\mu_B| \leq \Delta$). Although part (ii) involves more in-depth thought compared to regular NHST,  determining $\Delta$ is essentially equivalent to deciding the minimum effect size of scientific interest, as is commonly done in standard power analysis for sample size calculations.
Thus, while our approach attempts to keep
the operational simplicity of standard NHST, it directly meets researchers' needs.

\subsection{\ourapproach}
\label{sec:react_simple}

To better illustrate how our approach builds upon existing work, we first introduce it in a simplified version.

In its simplest form 
\ourapproach \ is composed by the following steps:
\begin{enumerate}
    \item \textbf{[Establish the null hypothesis]} Define a null hypothesis $H_0$ by establishing a pragmatic hypothesis, which is a range of values considered to be practically equivalent. For example, if $\mu_A$ and $\mu_B$ are the average effects of drugs A and B on a desired outcome, the pragmatic null may be $H_0: |\mu_A-\mu_B|\leq \Delta$, where $\Delta$ is the smallest difference of practical interest (also known as  the
smallest effect size of interest -- SESOI \citep{lakens2017equivalence}). In this case,  $H_0$ is usually called the equivalence range \citep{bauer1996unifying} or a region of practical equivalence (ROPE) \citep{kruschke2018rejecting} and is extensively used in equivalence testing.
The task of setting this region may be complex, but it essentially parallels the steps taken in standard power analysis, specifically in pinpointing the significant portions of the alternative hypothesis that need high power.  See Section \ref{sec:relatedWork} for ideas on how $\Delta$ can be chosen. 
Similarly, one may want to test 
    $H_0: |\rho|\leq \Delta$, where $\rho$ is the correlation coefficient between two quantities of interest, or $H_0:|\beta_i|\leq \Delta$, where $\beta_i$ is the coefficient of the $i$-th covariate in a regression. For simplicity, we now assume that the null hypothesis has the shape  
    $$H_0: |\phi|\leq \Delta$$
    for some parameter $\phi$, although our test is more general (see Section \ref{sec:methods}).
    \item \textbf{[Build a confidence set]} Create $C$, a confidence set
    for the parameter of interest, $\phi$. That is, $C$ contains values of $\phi$ that are consistent with the dataset that was observed.
\item \textbf{[Test $H_0$ using $C$]} Test the null hypothesis using the following three-way rule:
  \begin{equation*}
    \text{Decision}=
    \begin{cases}
      \text{\textbf{Accept} $H_0$} \ \text{if all values of $C$  are smaller than $\Delta$ in absolute value} \\
      \text{\textbf{Reject} $H_0$} \ \text{if all values of $C$  are larger than $\Delta$  in absolute value} \\
      \text{\textbf{Remain Agnostic}} \ \text{otherwise} \\
    \end{cases}
  \end{equation*}
\end{enumerate}
This procedure is illustrated in Figure \ref{fig::simple}.

\begin{figure}[ht]
    \centering
        \resizebox{.8\linewidth}{!}{\reactSimple}
    \caption{Illustration of \ourapproach\ to test  hypotheses  of the type $H_0: |\phi|\leq \Delta$. $C$ is a confidence set built using the data; the blue region represents the null hypothesis.}
    \label{fig::simple}
\end{figure}

\vspace{2mm}

Figure \ref{fig:pragmatic_camcog} shows an example of application of \ourapproach\ to the problem of investigating whether
CAMCOG scores can distinguish between three groups of patients: control (CG), mild cognitive impairment (MCI), and Alzheimer's disease (AD).
In each plot, the dashed line represents the precise hypothesis of interest, which states that CAMCOG scores are equally distributed among the compared groups. The blue region represents the null hypotheses associated with each pair of groups. We display the outcome of \ourapproach\ as the mean differences' confidence intervals between each pair of groups at a given sample size. In each comparison, we start by randomly sampling two observations from each group to derive the initial confidence interval. Then, in each step, we randomly add a new observation from one of the groups of interest and obtain a new confidence interval. For small sample sizes, the test remains agnostic on all three hypotheses. As the sample size increases, the pragmatic hypothesis for AD vs Control is rejected, the one for AD vs MCI is inconclusive, and the one for Control vs MCI is accepted. We obtain the same conclusions when changing the sorting order, with slight changes in MCI vs AD (Figure \ref{fig:camcog_resampling_plots} in  \proofs{}).
More details about this example can be found in Section \ref{sec:applications}.

 \begin{figure}[ht]
  \centering
  \includegraphics[width=0.9\textwidth]{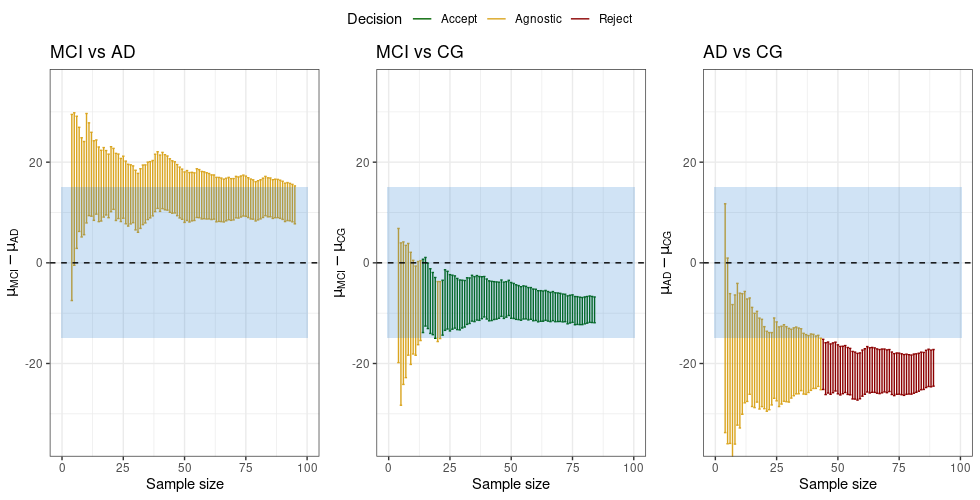}%
   \caption{Confidence intervals for
    the average difference between groups, $\mu_i-\mu_j$
    as a function of the sample size.
    The dashed line indicates the precise hypothesis 
    considered in each figure, $H_0: \mu_i = \mu_j$.
    The blue regions delimit
    the null hypotheses.}
  \label{fig:pragmatic_camcog}
 \end{figure}

\section{Connections to Existing Work}
\label{sec:relatedWork}

Next, we investigate the relationship between \ourapproach\ and similar approaches.
\vspace{2mm}

\paragraph{Power Analysis and Severity Tests.}
Power analysis is based on the following principle, suggested by J. Neyman \citep{neyman1957use}: 
\begin{quote}
``[If] the probability of detecting an appreciable error in the hypothesis
tested was large, say .95 or greater, then and only then is the decision in
favour of the hypothesis tested justifiable in the same sense as the decision
against this hypothesis is justifiable when an appropriate test rejects it at a
chosen level of significance."
\end{quote}
That is, according to this approach, if the power of the study is large (at some specific point of the alternative of interest, say $\P_{\phi=\Delta}(\text{reject})$ in the setting presented above), then a 
non-rejection can be interpreted  as acceptance (see also ``severe testing'' \citep{mayo2006severe,mayo2018statistical}---which uses a similar calculation to the p-value, though assuming a point of the alternative as true instead, to verify if one is warranted to accept the null hypothesis---for a related approach).
Thus, power analysis addresses issue 3 raised in the introduction by distinguishing between ``evidence of absence'' and ``absence of evidence''.
 Unfortunately, power analysis is often overlooked in favor of p-values, even when reported in the analysis. To address this issue, \ourapproach\ outputs a ``accept'', ``agnostic'' or ``reject'' decision, preventing any misinterpretation. 
\vspace{2mm}

\paragraph{Scheffé's method.}

Scheffé's method \citep{scheffe1999analysis} can be viewed as a special case of \ourapproach \ for an analysis of variance (ANOVA) model when testing hypotheses involving only linear combinations (contrasts) of treatment effects. Scheffé's method constructs the usual confidence ellipsoid for the treatment effects. To test the null hypothesis, it checks whether the line corresponding to the null hypothesis intersects the confidence ellipsoid. If such an intersection exists, the null hypothesis is not rejected; otherwise, it is rejected.

In contrast, \ourapproach \ allows consideration of a broader range of null hypotheses, including non-linear functions of treatment effects and pragmatic hypotheses derived from these functions. For pragmatic hypotheses, \ourapproach \ refines the "non-rejection" outcome into either "acceptance" or "remain undecided." Additionally, \ourapproach \ can be applied to various other models and is not restricted by the assumptions of ANOVA. See \autoref{sec:CAMCOG} for an application of \ourapproach \ that is similar to Scheffé's method when setting $\Delta = 0$, with the difference that \ourapproach \ is controlling for multiple testing.

\vspace{2mm}
\paragraph{Effect size estimates.}
Effect sizes are an attempt to address issues 1 and 2. They measure the strength  of a relationship
and are often used to complement NHSTs. 
These include the estimated coefficients $\hat \beta$ of the regression function \citep{kutner2005applied}, Cohen's d \citep{cohen2013statistical} and Pearson's r \citep{pearson1920notes}, among many others \citep{kirk2007effect,ellis2010essential,fritz2012effect}.   
If the null hypothesis is rejected, one can check whether the effect size is large enough to be practically significant. Although effect sizes can provide valuable information, they are often overlooked in favor of p-values.

To address these limitations, \ourapproach\ integrates both rejection of the null and practical significance into statistical inference. 
 Moreover, if one designs null hypotheses using effect sizes \citep{kruschke2018bayesian}, \ourapproach\ can become fully integrated with them. 
For instance, Cohen's d between two groups can be expressed as $\phi:=\sigma^{-1}(\mu_A-\mu_B)$, where $\sigma$ is the standard deviation of the response variable. The equivalence region  can be stated in terms of $\phi$ as $H_0: |\phi| \leq \Delta$. This hypothesis can then be tested using our approach via the confidence set for $\phi$, $d\pm \text{SE}_d \times t_{n_A+n_B-2}$, where $d = s^{-1}(\bar{x}_1 - \bar{x}_2)$ is an estimate of $\phi$, $t$ is the Student's t-distribution,  $n_A$ and $n_B$ are the sample sizes for each group, and $\text{SE}_d$ is an estimate of the standard error of $d$ \citep{goulet2018review}. 

We note, however, that applying a test based on Cohen's d may not be ideal. As \citet{Lakens2022} puts it,
\begin{quote}
``Setting them [equivalence bounds] in terms of Cohen’s d leads to bias in the statistical test, as the observed standard deviation has to be used to translate the specified Cohen’s d into a raw effect size for the equivalence test (...) [A]s equivalence testing becomes more popular, and fields establish smallest effect sizes of interest, they should do so in raw effect size differences, not in standardized effect size differences."
\end{quote}

\paragraph{Equivalence Range.}
The idea of working with regions of practical equivalence (instead of  point null hypotheses) to address issues 1 and 2
has been presented under many different names---such as good-enough belt \citep{keren1993}, equivalence range \citep{bauer1996unifying}, indifference zone \citep{hobbs2007}, effective null set \citep{gross2014}, and even the ROPE \citep{kruschke2018rejecting} itself---but a definition with sufficient generality to cover for more complex settings is more recent (pragmatic hypothesis, \citet{Esteves2019}). Let $\phi_0$ be the point from which we wish to derive the region of equivalence (that is, we are interested in creating a region of equivalence around $H_0:\phi=\phi_0$), $\Theta$ be the parameter space and $d(\cdot, \cdot)$ be a dissimilarity function. Then, the pragmatic hypothesis will be represented by
$$
    Pg(\phi_0, d, \Delta) := \{\phi \in \Theta: d(\phi_0, \phi) \le \Delta\}.
$$
The setting described in Section \ref{sec:react_simple} corresponds to choosing $d(\phi_0, \phi) = |\phi_0 - \phi|$ and $\phi_0 = 0$. 
There are many heuristics for choosing $\Delta$ \citep{lakens2017equivalence, Lakens2022, Wang2023, lassance2024}. These include:
\begin{itemize}
    \item Relating $\Delta$ to another quantity in the literature about which it is easier to obtain intuition; we use this in \autoref{sec:meta}.
    \item Identifying positive results in the literature that evaluate the same (or a similar) effect as your own study and choosing the smallest $\Delta$ such that \ourapproach{} would lead to accepting the hypothesis in these previous cases. While this strategy may downplay random variability if few studies are used to derive $\Delta$, it acts as a starting point from which researchers can propose changes later. Such an approach is particularly useful in the context of reproducibility studies since rejecting $H_0$ in the new study based on criteria that would have accepted the same hypothesis in the old one can be interpreted as a failure to reproduce the original finding.
    \item Setting $\Delta$ as the smallest change such that patients report
    an improvement from their original conditions. This can routinely be obtained through the use of patient reported outcome measure (PROM) scores. Even when the perception of improvement varies substantially between patients, there is a selection process available that ensures the optimality of the selected $\Delta$ \citep{Wang2023}.
    \item  
    Setting $d(\phi_0,\phi)$ to be a measure of effect size (such as the standardized mean difference between two populations) and taking $\Delta$ to be the smallest effect size that is practically significant. This approach parallels traditional sample size calculations performed through power analysis.
\end{itemize}

When there are no clear strategies for determining $\Delta$, there is available software that can help provide an automated suggestion \citep{makowski2019}. \ourapproach\ can be used within any of these approaches.

\paragraph{Equivalence Tests, TOST and B-values.}
Equivalence tests were originally developed to compare the bioequivalence of two drugs to address issues 1 and partially issue 3. Rather than testing the hypothesis of ``no effect'' ($\phi=0$), the hypothesis is modified to `` practically no effect'' ($|\phi|< \Delta$) \citep{westlake1976symmetrical,schuirmann1987comparison}. Additionally, the null hypothesis is set as the hypothesis of a practical effect, namely $H_0: |\phi|> \Delta$. Thus, under a  Neyman-Pearson interpretation of the test outcomes, rejecting the null hypothesis would lead to the conclusion that there is  absence of practical effect, $|\phi|\leq \Delta$. This makes equivalence tests very useful in various fields including political science \citep{rainey2014}, communication research \citep{weber2012testing}, anthropology \citep{smith2020p}, sensory science \citep{meyners2012equivalence}, psychology \citep{lakens2018equivalence}, and clinical trials \citep{walker2011understanding, wellek2012establishing, friedman2015fundamentals, leung2020non}.

A popular method for conducting equivalence tests  is the Two One-Sided Tests (TOST \citep{schuirmann1987comparison}). The standard TOST procedure for two-sample testing involves (i) constructing a $(1-2\alpha)$-level confidence interval for $\mu_A-\mu_B$, and (ii) rejecting the null hypothesis that  indicates the absence of a practical effect, $H_0: |\mu_A-\mu_B|> \Delta$,  only if the confidence interval falls entirely within the interval $(-\Delta,\Delta)$. While this procedure is similar to \ourapproach\ in that it tests for equivalence using a confidence set, it only yields two possible outcomes: either reject the null hypothesis and conclude the absence of effect, or fail to reject the null hypothesis. 
In fact, TOST corresponds to applying \ourapproach\ taking $C$ to be a $(1-2\alpha)$-level  confidence interval 
and merging the decisions ``accept'' and ``agnostic''.
Thus, \ourapproach\ allows for a more nuanced approach by also permitting acceptance of the hypothesis of a practical effect.
 
 Unlike \ourapproach, TOST procedures require tailoring to address each specific problem:
 utilizing any $(1-2\alpha)$-level confidence set does not necessarily control type I error rates at $\alpha$ \citep{wellek2010testing}. Indeed, \citet{berger1996bioequivalence} show several examples of TOST tests that do not control type I error rates \citep{berger1996bioequivalence}. Some examples of TOST methods include regression \citep{dixon2005statistical, robinson2005regression,campbell2020equivalence,alter2021determining}, parametric and non-parametric paired sample tests \citep{mara2012paired}, and correlation coefficients \citep{counsell2015equivalence}.
\ourapproach, however, controls type I error at $\alpha$ as long as $C(\D)$ is a $(1-\alpha)$-level confidence set (Proposition \ref{proper:control}).
Furthermore, the null and alternative hypotheses can be interchanged without affecting the conclusions of the test: \ourapproach\ is symmetric (Proposition \ref{prop:symmetry}).
\ourapproach{} can also be easily used for hypotheses that involve several parameters (see Section \ref{sec:CAMCOG} for an example), while TOST requires further work (see e.g. \citet{yang2015multigroup} for the development of TOST to compare multiple groups).

Recently, a two-stage equivalence testing procedure that has the advantage of deriving a region of equivalence solely based on data, the empirical equivalence bound (EEB), has been proposed \citep{zhao2022b}.
Let $[L_0, U_0]$ and $[L,U]$ respectively be the symmetrical $(1-\alpha)$ and $(1-2\alpha)$ confidence intervals for $\phi$ and set $B := \max\{|L|,|U|\}$ (the largest deviation from 0 of the interval). The EEB is given by
\begin{equation*}
    EEB_\alpha(\beta|C) = \inf_{b \in [0, \infty]}\{b: F_B(b|C, \phi = 0) \ge \beta\},
\end{equation*}
where $C$ is either $0 \in [L_0, U_0]$ or $0 \notin [L_0, U_0]$, $F_B$ is the cumulative distribution of $B$ and $\beta$ is a fixed probability. Therefore, $[-EEB,EEB]$ is the smallest symmetrical region such that one would reject $\phi = 0$ with probability $\beta$.

This serves as yet another suggestion for deriving $\Delta$ and, once its value is fixed, \ourapproach{} reaches the exact same conclusions as the two-stage testing procedure in \citet{zhao2022b} when using the $100(1-2\alpha)\%$ confidence set for testing.

\paragraph{Non-inferiority and Superiority Tests.}
\ourapproach\ is closer in spirit to tests that combine superiority, non-inferiority and equivalence tests, such as \citet{julious2004sample}  and other variations \citep{tryon2001evaluating,goeman2010three,friedman2015fundamentals,zhao2016considering,lakens2018equivalence,zhao2022b}. While these approaches are highly informative and gaining  in popularity, they are specific to certain tests and hypotheses, and require tailoring to fit each individual problem. Therefore, there is no guarantee that they will control type I error at $\alpha$ \citep{berger1996bioequivalence} 
 for any given problem (indeed, they are not designed to have type I error control globally). In Section \ref{sec:applications}, we provide examples of scenarios where it may be difficult to adapt such an approach, whereas our approach, \ourapproach, remains user-friendly. Moreover, as shown in Proposition \ref{prop:fwer}, \ourapproach\ automatically controls the  Family-Wise Error Rate \citep{lehmann1957theory,wang1999multiple,gupta1981multiple}  (FWER); there is no need to use procedures such as Bonferroni correction to account for multiple testing.


\paragraph{HDI+ROPE, GFBST and S-values.}
The HDI+ROPE (Highest Density Interval+Region of Practical Equivalence) \citep{kruschke2010bayesian, kruschke2018rejecting,keysers2020using} represents a specific instance of \ourapproach. Indeed, its definition (see e.g. \citet[page 291]{kruschke2010bayesian}) directly corresponds to  \ourapproach\ with the confidence set $C$ as the $(1-\alpha)$-level Bayesian Highest Density Credible Interval.
In a similar vein, the GFBST \citep{Stern2017} (Generalized Full Bayesian Significance Test) corresponds to choosing $C$ as the $(1-\alpha)$-level Bayesian Highest Posterior Density Region (HPD), given by
$$C(\D)=\{\theta \in \Theta: f(\theta|\D) \geq t_{1-\alpha}\},$$
where $t_{1-\alpha}$ is chosen so that $\P(\theta \in C(\D)|\D)=1-\alpha$, and $f(\theta|\D)$ is the posterior distribution of $\theta$ given data $\D$.
  New theoretical guarantees regarding the average coverage probabilities of such procedures are provided in Section \ref{subsec:properties} and in the \proofs{}
(Theorem \ref{thm:family-wise-bayes}). Furthermore, the approach presented by \citet{patriota2013classical} can also be viewed as a specific instance of \ourapproach, where the sets $C$ are constructed using s-values.
\vspace{2mm}

\paragraph{Three-way hypothesis tests.}
 Three-way tests provide a more  nuanced approach to hypothesis testing \citep{jones2000sensible,rice2023three} and were suggested \emph{en passant}
by J. Neyman \citep{Neyman1976}, who argued that  ``The phrase `do not reject H' is longish and cumbersome ... (and) should be distinguished from a `three-decision problem' (in
which the) actions are: (a) accept H, (b) reject H, and (c) remain in doubt.''
Neyman however did not develop this approach.
To the best of our knowledge, \citet{hays1963statistics} was the first one to present a formal 
three-way approach to hypothesis testing under a decision-theoretic framework. 

Three-way hypothesis testing has several benefits. It can address concerns about replication and the limited publication of null results \citep{campbell2018conditional} and is essential to the evolution of scientific theories \citep{Esteves2019}.
Recent research has explored the many advantages of three-way tests. For example, \citet{berg2004no} showed that three-way tests can control both types I and II error probabilities in a setting of simple versus simple hypotheses. This is in contrast to two-way tests, where only one error can be controlled. Subsequently, \citet{Coscrato2018} generalized this approach to composite alternative hypotheses. 
\ourapproach\  controls both type I and type II error probabilities to be no more than the same level $\alpha$ (Property \ref{proper:control}).
In addition, \citet{Esteves2016,Stern2017} and \citet{esteves2023logical} showed that
three-way tests can address logical inconsistencies that occur in standard two-way tests \citep{izbicki2012testing,da2015bayesian,Izbicki2015,fossaluza2017coherent}. 
This is one of the arguments made by \citet{kruschke2018rejecting} to justify why HDI+ROPE should be preferred over TOST+NHST.

A three-way testing approach that is very closely related to \ourapproach{} is the PASS-test \citep{gross2014}, which uses a region of equivalence combined with confidence intervals to reach decisions. The main differences between the methods is that \ourapproach{} is more general (it can derive regions of equivalence and tests to combinations of parameters, instead of only uniparametric hypotheses) and 
has a series of useful properties particular to it (see Section \ref{sec:methods}). One of them (logical coherence, Property \ref{proper:coherence}) guarantees that no contradictory conclusions are reached when REACT is applied to multiple tests. The main advantage of using three-way tests such as these is that, since they are based on intervals instead of p-values, they provide greater intuitive appeal  than p-value-based tests \citep{rainey2014}.

\vspace{2mm}

\paragraph{Tests based on posterior probabilities and Bayes Factors.}
 An alternative approach to testing hypotheses is the Bayesian framework. Under this perspective, one typically computes either the posterior probability of the null hypothesis, $\P(H_0|\D)$, or the Bayes Factor, $\P(\D|H_0)/\P(\D|H_1)$, and rejects the null if the values are small \citep{Berger2013}. This approach solves many of the issues associated with NHST. For instance, Bayesian tests can accept the null hypothesis \citep{rouder2009bayesian,kelter2020bayesian}.
 
 Furthermore, within a Bayesian framework, it is customary to formulate null hypotheses that are not precise \citep{Edwards1963, good2009, berger1985}, thereby circumventing the challenge of exclusively dealing with implausible hypotheses \citep{schervish2012theory}. This, however, remains mostly restricted to uniparametric hypotheses, such as $H_0: |\phi - \phi_0| \in [\delta_L, \delta_U]$, where $\delta_L \le 0 \le \delta_U$ are known beforehand \citep{hobbs2007, kruschke2018rejecting}. The more popular alternative remains to assign probability masses to precise hypotheses \citep{jeffreys1961, kass1993, migon2014}, done mostly due to practical reasons instead of a true representation of the researcher's beliefs.
 
 Moreover, tests based both on posterior probabilities and on Bayes factors are not fully logically coherent in the sense described in Section  \ref{sec:methods} \citep{lavine1999bayes,Izbicki2015}. \ourapproach\ can be used within a Bayesian context, although it is not based on computing posterior probabilities of the null hypothesis (see Section \ref{sec:methods}).

\vspace{2mm}

\section{\ourapproach\ and its properties}
\label{sec:methods}

Next, we introduce the general version of \ourapproach. Our goal is to test
one or more hypotheses regarding
parameters $\theta$,
which assume values in $\Theta$,
a subset of $\Re^d$.
In this context, a null hypothesis is
of the form:
$H_0: \theta \in \Tz$,
where $\Tz$ is a subset of $\Theta$ 
and the alternative hypothesis, $H_1$,
is $H_1: \theta \in \Tzc$.
Whenever there is no ambiguity,
$\Tz$ is used instead of $H_0$.
 All definitions and 
proofs of the properties stated in this section  are
found in  the \proofs{}.

In order to test $H_0$,
\ourapproach \ requires one
to construct  a region of
 values for $\theta$.
One possible region is obtained by
using data, $\D$,
to construct a confidence region
for $\theta$, $C(\D)$.
If a frequentist approach is used, then
one requires that
\begin{align*}
 \P_\theta\left(
 \theta \in C(\D)\right)
 =1-\alpha, \ \text{for all } \theta \in \Theta,
\end{align*}
where $1-\alpha$ is the 
confidence level of $C(\D)$. Notice that the randomness of this probability is on the data $\D$; once the data is observed there is no randomness left.
For a Bayesian implementation of 
\ourapproach \ one could use
a $(1-\alpha)$-level credible region \citep{Berger2013}, that is, a region such that
$$\P\left(
 \theta \in C(\D)|\D \right)
 =1-\alpha,$$
where the probability's randomness stems from the posterior distribution of $\theta$ given the data $\D$.

\ourapproach\ tests $H_0$
using the following rule
(illustrated in Figure \ref{fig::region}):
  \begin{equation*}
    \text{Decision}=
    \begin{cases}
      \text{\textbf{Accept} $H_0$} \ \text{if $C(\D) \subseteq \Theta_0$,} \\
      \text{\textbf{Reject} $H_0$} \ \text{if $C(\D) \subseteq \Theta^c_0$,} \\
      \text{\textbf{Remain Agnostic}} \ \text{if $C(\D)$ intersects with both $\Theta_0$ and $\Theta^c_0$}. \\
    \end{cases}
  \end{equation*}

\begin{figure}[!http]
 \centering
 \includegraphics[width = \textwidth]{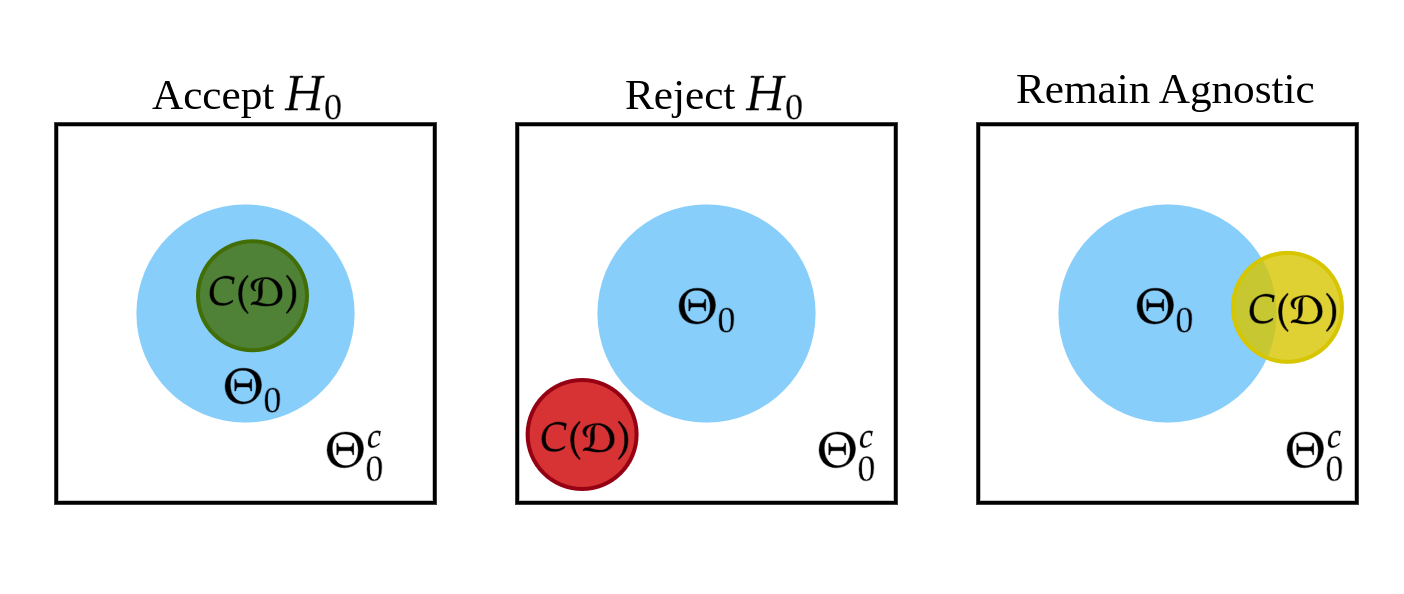}
 \caption{Illustration of \ourapproach \ to test a generic hypothesis  $H_0: \theta \in \Tz$. $C(\D)$ is a confidence set built using  data $\D$; the blue region represents the null hypothesis. Whenever $C(\D)$ is a subset of $\Tz$
 (left), all plausible values for $\theta$ 
 are in $\Tz$. Hence, $\Tz$ is accepted.
 Whenever $C(\D)$ lies outside of $\Tz$
 (middle), no plausible value for $\theta$
 is in $\Tz$. Therefore, $\Tz$ is rejected.
 Finally, if some points of $C(\D)$ are
 inside $\Tz$ and others are outside,
 then some plausible points for $\theta$
 are compatible with $\Tz$ and
 others are incompatible. In this situation,
 $\Tz$ remains undecided.}
 \label{fig::region}
\end{figure}

In what follows,  \fourapproach \ denotes a \ourapproach \ procedure based on a $(1-\alpha)$-level frequentist region, while  \bourapproach \ denotes a \ourapproach \ procedure based on a $(1-\alpha)$-level Bayesian credible region.  We use \ourapproach\ to refer to scenarios where the property in question holds for any version of $C(\D)$.

\fourapproach \ can
also be obtained using standard 
p-values for point null hypotheses:

\begin{property} \textbf{[Computation of \fourapproach\ using p-values]}
\label{proper:pvals}
    Let $\text{p-val}_{\D}(\theta_0)$ be a  p-value for the hypothesis $H_0:\theta=\theta_0$.
    The procedure
    \begin{equation*}
    \begin{cases}
      \text{\textbf{Accept} $H_0$} \ \text{if $\max_{\theta \in \Theta_0^c} \text{p-val}_{\D}(\theta) \leq \alpha$,} \\
      \text{\textbf{Reject} $H_0$} \ \text{if $\max_{\theta \in \Theta_0}\text{p-val}_{\D}(\theta) \leq \alpha$,} \\
      \text{\textbf{Remain Agnostic}} \ \text{otherwise.} \\
    \end{cases}
  \end{equation*}
  \vspace{6mm}
   is equivalent to
  calculating \fourapproach \ through
  the (1-$\alpha$)-level confidence set:
  $C(\D):=\left\{\theta \in \Theta: \text{p-val}_{\D}(\theta)  > \alpha \right\}.$
\end{property}

Oftentimes, the parameter space can be decomposed as $\Theta=\Phi \times \Psi$, where $\phi \in \Phi$ are parameters that will be tested, and $\psi \in \Psi$ are nuisance parameters. Theorem \ref{thm:nuisance} in  the \proofs{}  shows that, as long as a confidence (or credible) set for the parameters of interest $\phi$ is available, \ourapproach\ can easily handle nuisance parameters. Moreover, Theorem \ref{thm:pvals_nuisance} shows how \fourapproach \ can be computed using p-values in this setting, and Theorem \ref{thm:evals} shows its Bayesian counterpart.

In many problems, one is interested in testing different hypotheses, each one dealing with a different set of parameters.
 For example,
in a linear regression model where $\beta_0+\beta_1x_1+\ldots+\beta_n x_d$, one is often interested in testing $H_0: \beta_i=0$ for each $i$ at a time. 
\ourapproach\ for such hypotheses can be obtained at a low computational cost 
as long as a confidence set for all parameters is available:

\begin{property} \textbf{[Easy computation of \ourapproach\ for low-dimensional hypotheses]}
\label{prop:projection-equivalence}
    Assume that the parameter space can be decomposed as $\Theta=\Phi \times \Psi$ and let $H_0: \phi \in \Phi_0$, $\Phi_0 \subset \Phi$, be the hypothesis of interest.
    Also, let $C(\mathcal{D})$ be a region estimator on $\Theta$ and 
    $C_\Phi(\mathcal{D})$ be the projection of the confidence region $C$ on the parameters $\phi$, that is, $$C_\Phi(\mathcal{D}) := \{\phi \in \Phi: \exists \psi \in \Psi \text{ such that } (\phi, \psi) \in C(\mathcal{D})\}.$$ Then, \ourapproach \ can be computed via
    \begin{equation*}
        \text{Decision}=
        \begin{cases}
            \text{\textbf{Accept} $H_0$} \ \text{if $C_\Phi(\D) \subseteq \Phi_0$,} \\
            \text{\textbf{Reject} $H_0$} \ \text{if $C_\Phi(\D) \subseteq \Phi^c_0$,} \\
            \text{\textbf{Remain Agnostic}} \ \text{otherwise.} \\
        \end{cases}
  \end{equation*}
  In words, \ourapproach\ only requires evaluating whether the projection of the confidence set belongs to the null hypothesis.
\end{property}

In the context of multiple testing, Theorem \ref{thm:nuisance} in the Appendix ensures that one only needs to build a confidence region for the parameters of interest, while Property \ref{prop:projection-equivalence} guarantees that using projections of such region will not affect the conclusions or harm the properties of \ourapproach{}. If there are nuisance parameters, using Property \ref{prop:projection-equivalence} directly on $C_\Theta(\mathcal{D})$ is not advised, as projections often lead to tests with a lower level of significance than the nominal one, potentially having lower statistical power.

The importance of Property \ref{prop:projection-equivalence} is illustrated in
Section \ref{sec:CAMCOG}, where we show how to implement \ourapproach \ for multiple pairwise comparisons.

\ourapproach \ has
desirable statistical and 
logical properties which we explore in what follows.

\subsection{Statistical properties}
\label{subsec:properties}

Whenever $C(\D)$ is a confidence region,
\ourapproach \ has
desirable frequentist statistical properties.
For instance, \fourapproach \
controls not only the type I error rate  but
also the type II error rate:

\begin{property} \textbf{[Type I and Type II error rate control]}
\label{proper:control}
\fourapproach \ has  Type I and Type II errors rates of at most $\alpha$.
\end{property}

Although the control of type I error rate is
a standard criterion for
hypothesis tests at point nulls,
such control is seldom obtained
for the type II error rates.
\fourapproach \ obtains the latter by
remaining undecided when
the data is not 
sufficiently informative about
the null hypothesis.

Next, we show that when
  conducting multiple hypothesis tests, \fourapproach \  controls a strong version of family-wise error rate (FWER), in the sense that it controls the probability of making at least one Type I or Type II error.
This is more stringent than the standard notion of FWER, which typically focuses only on controlling Type I errors.

\begin{property}\textbf{[Family-wise error rate (FWER) control]} 
 \label{prop:fwer}
 If several hypotheses are tested, \fourapproach \ controls a generalized version of the  Family-Wise Error Rate (FWER \citep{tukey1953problem}) over 
 multiple hypothesis tests. 
 That is, for every $\theta \in \Theta$,
 \begin{align*}
  \begin{cases} 
   \text{FWER}_{I} :=
   \P_\theta(\text{At least one correct hypothesis is rejected}) \leq \alpha, 
   & \\ \text{FWER}_{II} :=
   \P_\theta(\text{At least one false hypothesis is accepted}) 
   \leq \alpha, & \\
   \P_\theta(\text{At least one correct
   hypothesis is rejected or
   one false hypothesis is accepted})
   \leq \alpha. &
  \end{cases}
 \end{align*}
\end{property}
Remarkably, this control is achieved  without necessitating the correction of significance levels for multiple testing.
See Theorem \ref{thm:family-wise-bayes}, in the \proofs{}, for a Bayesian counterpart of this result using \bourapproach.

Next, we show that
\ourapproach \ is  typically consistent in the sense that,  as the sample size increases,
\ourapproach \ 
rejects each hypothesis that is false and
accepts each one that is true.
Thus, \ourapproach\ does not have 
the issue in which the null hypothesis
is rejected as the sample size increases.

\begin{property} \textbf{[Consistency as n increases]}
 \label{proper:consistency}
 If the confidence set $C(\mathcal{D})$ converges to the true value of the parameter (fixed), $\theta$, (see Definition \ref{def:convergeSet}),   and $\theta$ does not lie on the boundary of the null hypothesis $\Theta_0$, then,
 as the sample size increases, the probability that 
 \ourapproach\ accepts $H_0$ 
 when it is true goes to one, and the probability that it rejects  $H_0$ goes to zero. This property is illustrated in Figure \ref{fig:pragmatic_camcog}.
\end{property}
 Besides the above properties,
one might also be interested in
the ``best'' confidence set to
be used in step 2 of \ourapproach .
In the following, we discuss
a first development of theory
for answering this question.

As a first challenge, one must answer
what is a ``best'' test in the context of \ourapproach .
Since \fourapproach \ controls both type I and type II errors, a possible generalization of standard theory is
to seek a test that minimizes the type III error, that is,
the probability of remaining undecided.
However, it can be hard to find such a test that
minimizes this error uniformly among
all parameter values and hypotheses that can be tested.

The following property provides a first answer when
one restricts attention to interval confidence sets and
unilateral and bilateral hypotheses. A formal mathematical development is presented in Appendix \ref{sec::addTheorems}.
 \begin{property}
  Consider that one wishes to test solely
  unilateral and bilateral hypotheses regarding
  the parameters of interest. Also, consider that
  there exists a standard uniformly most powerful (UMPU) test for each point hypothesis.
  If \fourapproach \ uses the confidence set obtained
  from inverting the UMPU tests, then the probability that
  it remains undecided is uniformly lower than that of
  every region test based on an interval confidence set that has the same type I and type II error rates.
 \end{property}

Finally, under a Bayesian perspective, accepted hypotheses have high posterior probability, while rejected hypotheses have low posterior probability:

\begin{property}\textbf{[Posterior probability of the null hypothesis]} 
\label{prop:bayes_prob}
If  \bourapproach\ accepts $H_0$, then the posterior probability of $H_0$ is larger than $1-\alpha$. Moreover, if 
\bourapproach\ rejects $H_0$, then the posterior probability of $H_0$ is smaller than $\alpha$.
\end{property}

Notice however that the reverse is not true: \bourapproach\ may not accept a hypothesis with large posterior probability. One example is when \bourapproach\ remains agnostic because at least one element of $H_0$ does not intersect with the credible region, due to it residing in an area with very small posterior probability.

\subsection{Logical properties}

Besides having good statistical properties,
\ourapproach \ is also coherent from a
logical perspective:

\begin{property} \textbf{[Logical coherence]} 
 \label{proper:coherence} 
 \ourapproach\ is logically coherent in the sense described in Definition \ref{def:coherence}. 
\end{property}

That is, if one treats
accepted hypotheses as true and 
rejected hypotheses as false, then
no logical contradiction is obtained.
For example:
\begin{itemize}
 \item If $H_0: |\mu_A - \mu_B| < \Delta$ 
 is accepted, then so is
 $H_0: |\mu_A - \mu_B| \leq \Delta$ .
 \item If both $H_0: \mu_A - \mu_B \leq 1$ 
 and $H_0: \mu_A - \mu_B \geq 0$ are 
 accepted, then
 $H_0: 0 \leq \mu_A - \mu_B \leq 1$ will be as well.
 \item  
 Consider pairwise null hypotheses of the form $H_0^{i,j}:|\mu_i-\mu_j| \leq \Delta$, where $\mu_i$'s are parameters of the model, and a global null hypothesis
$H_0:\max_{i,j} \{ |\mu_i-\mu_j| \} \leq \Delta$.
  If all  $H_0^{i,j}$'s are accepted, so is $H_0$. 
  Similarly, if at least one  $H_0^{i,j}$ is rejected, so is  $H_0$.
 \end{itemize}

Most statistical procedures do not satisfy this property, which makes their outcomes hard to interpret \citep{Izbicki2015}. The reader is referred to \citet{hansen2023coherent} for other types of coherence.

A consequence of Property \ref{proper:coherence} is that,
unlike standard tests, where it is crucial to determine which hypothesis is labeled as the ``null'' and which one is labeled as the ``alternative,'' \ourapproach\ is indifferent to this choice.
In other words,
 the choice of whether the null hypothesis is $H_0: |\mu_A - \mu_B| \leq \Delta$ or $H_0: |\mu_A - \mu_B| > \Delta$ 
     does not materially affect the conclusions. If the former is rejected, it means that the latter is accepted and vice versa:

\begin{property} \textbf{[No need to specify which one is the null hypothesis]}
\label{prop:symmetry}
The null and the alternative hypotheses can be exchanged without materially affecting conclusions, and therefore there is no need to distinguish the labels ``null'' and ``alternative'' hypothesis.
\end{property}

\section{Applications}
\label{sec:applications}

\subsection{Cambridge Cognition Examination}
\label{sec:CAMCOG}
The Cambridge Cognition Examination \citep{Roth1986} (CAMCOG) is a widely-used questionnaire for measuring the extent of dementia and assessing the level of cognitive impairment. We analyze data from \citet{Cecato2016} to investigate whether CAMCOG scores can distinguish between three groups of patients: control (CG), mild cognitive impairment (MCI), and Alzheimer's disease (AD).
We assume that the score of the $k$-th patient in group $i$ is given by $Y_{i,k} = \mu_i + \epsilon_{i,k}$, where $\mu_i$ is the population average for group $i$ and $\epsilon_{i,k}$ is a Gaussian random variable with mean 0 and variance $\sigma_{i}^2$.
Our analysis focuses on testing the hypothesis
\begin{equation}
    \label{eq:hyp_pair}
    H_0: |\mu_i - \mu_j| \leq  \Delta
\end{equation}
for all pairs of groups. 

An initial approach to address this problem involves constructing confidence sets for individual parameters, specifically $\phi:=|\mu_i - \mu_j|$. The process of building these confidence sets is illustrated in \autoref{fig:pragmatic_camcog}, where we explore how the sample size influences the test results. The confidence intervals are constructed using the standard confidence set for the difference between two means, assuming a Gaussian distribution \citep{diez2012openintro}, where each difference may have a different variance.

The value of $\Delta$ was determined using the classification dissimilarity introduced by \citet{Esteves2019}:
When the pragmatic hypothesis does
not hold, there exists a classifier based on the CAMCOG score which is able to discriminate the groups under comparison with an accuracy (that is, a probability of predicting the correct class) of at least 80\%. 

Although the previous approach is sensible, it is technically not a fully-\ourapproach\ approach: \ourapproach\ requires a single confidence set $C(\mathcal{D})$ to be used to test all hypotheses. A full-\ourapproach\ can be obtained by using the following $100(1-\alpha)\%$  confidence region 
for the vector of means $\boldsymbol{\mu} = (\mu_{CG}, \mu_{MCI}, \mu_{AD})$ \citep{johnson2002,meeker1995}: 
\begin{equation}
    \label{eq:confidence_ellipsoid}
    (\overline{\boldsymbol{x}} - \boldsymbol{\mu})'\hat{\Sigma}_{\overline{\boldsymbol{x}}}^{-1}(\overline{\boldsymbol{x}} - \boldsymbol{\mu}) \le \chi^{2}_{p}(\alpha),
\end{equation}
where $\overline{\boldsymbol{x}}$ is the sample mean, $\hat{\Sigma}_{\overline{\boldsymbol{x}}}$ is the estimate of the covariance matrix of $\overline{\boldsymbol{X}}$, $p$ is the size of $\boldsymbol{\mu}$ and $\chi^2_{\cdot}(\cdot)$ is the quantile function of the $\chi^2$-distribution. Moreover, $\hat{\Sigma}_{\overline{\boldsymbol{x}}}$ is a diagonal matrix in this case because the data on each subject was collected independently.

Once the confidence region has been derived, any hypothesis can be tested through the \ourapproach{} framework by checking if the confidence region is contained in $H_0$. While evaluating if $C(\mathcal{D})$ is contained in $H_0$ is feasible, Property \ref{prop:projection-equivalence} ensures that the same conclusions can be reached by using the projection (in this case, the ellipses of \autoref{fig:camcog_react_ellipse}) of $C(\D)$ on $(\mu_i,\mu_j)$, requiring fewer calculations than using $C(\mathcal{D})$ and allowing for graphical visualization of the results in two dimensions. By projection, we mean the set
$$
  \left\{  (\mu_i, \mu_j) \in \mathbb{R}^2:  \exists \mu_k \in \mathbb{R} \text{ with } (\mu_i, \mu_j, \mu_k) \in C(\D) \right\}.
$$

In \autoref{fig:camcog_react_ellipse}, the blue area represents the region of equivalence and each ellipse is a confidence region whose color implies the result of the test (yellow for remaining undecided, green for acceptance and red for rejection). While there is not enough evidence to conclude if MCI patients can be discriminated from those with AD (absence of evidence of a practical effect), the CAMCOG is unable to differentiate between MCI and the control group (evidence of absence of a practical effect). This is not the case when comparing AD and CG groups, supporting the idea that the CAMCOG is able to differentiate between them (evidence of presence).
Finally, Property \ref{proper:coherence} 
and the fact that $|\mu_{AD} - \mu_{CG}|\ \leq \Delta$ was rejected implies that the hypothesis of no relevant difference between all  groups ($H_0: \max \left\{|\mu_{MCI} - \mu_{AD}|, |\mu_{MCI} - \mu_{CG}|, |\mu_{AD} - \mu_{CG}| \right\} \leq \Delta$) is also rejected. A  Bayesian version of this analysis is presented in \autoref{fig:camcog_bayes} and leads to the same conclusion.


\begin{figure}[ht]
\centering   
\includegraphics[width = .9\textwidth]{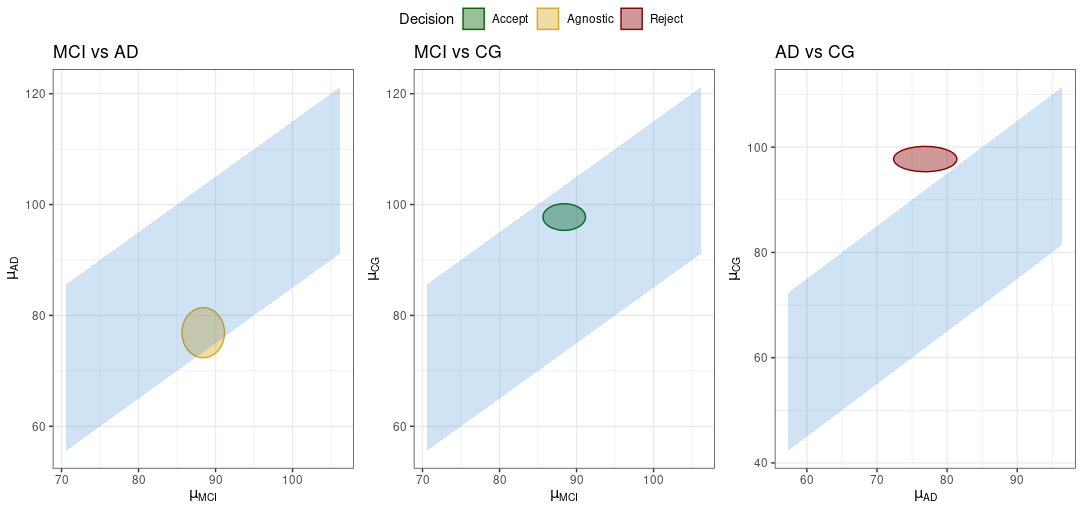}
    \caption{Pairwise group comparisons given by \ourapproach\ for the Cambridge Cognition Examination example. The ellipses are projections of the multivariate confidence set of Equation \ref{eq:confidence_ellipsoid}  on each pair of parameters $(\mu_i,\mu_j)$, while the blue regions represent the null hypotheses $|\mu_i-\mu_j| \leq \Delta$ for the various groups. Property \ref{prop:projection-equivalence} guarantees it is enough to evaluate these projections to test these hypotheses.}
    \label{fig:camcog_react_ellipse}
\end{figure}

\vspace{2mm}

\subsection{Meta-analysis}
\label{sec:meta}
This application is an adaptation of the meta-analysis originally performed by \citet{Teixeira2022} and uses the Number Necessary to Treat (NNT) as a means to obtain the desired region of equivalence for the risk differences. The NNT is the expectation of the least people required for the treatment to present better results than the control and thus is commonly used in the literature as an indicator of clinical significance and its value can be easily provided \citep{Citrome2011}.

The original study evaluates patients' adherence to tobacco cessation protocols comparing traditional approaches combined with computer-assisted health technologies to traditional approaches themselves. The treatments are compared by two main outcomes: the adherence to the follow-up period of treatments without any drug (labeled as ``Follow-up'') and the adherence to the pharmacotherapy, on studies that used any drug besides nicotine replacement (labeled ``Pharmacotherapy''). Since the study evaluates the adherence to treatment, not the efficacy of the treatment itself, the NNT of interest should be at least smaller than the best treatments recommended by the literature \citep{Cahill2016, Stead2012}, which leads to $\text{NNT} < 11$.

\autoref{fig:nnt_forestplot} presents the forest plot for both treatments, with the blue area representing the region of equivalence substituting the hypothesis that the risk difference is less or equal to 0 (treatment no better than control), the X-axis providing the confidence intervals obtained from each study and the ``pooled'' label representing the interval that results from aggregating all studies. The aggregation was done either through \textbf{(i)} a Random Effects Model (REM) with random intercepts \citep{bakbergenuly2018}---which mirrors the pooling strategy of the original study---on the risk difference of each study or \textbf{(ii)} a fixed effects model, which assumes that the samples from all studies come from exactly the same population.
Unlike \citet{Teixeira2022}, in this case, the difference between the probabilities of success of treatment and control was used instead of the relative risk. This was due to the fact that the NNT is the inverse of such a difference, so one can be directly translated into the other. Since the interest is to evaluate clinical significance, the region of equivalence chosen was $[-1, 1/6]$, meaning that the NNT of the study has to be less than 6 for its results to be practically significant (treatment better than control). A Bayesian version of this analysis is presented in \autoref{fig:forestplot_bayes} and leads to the same conclusion.

\begin{figure}[h!]
    \centering
    \includegraphics[width = .85\textwidth]{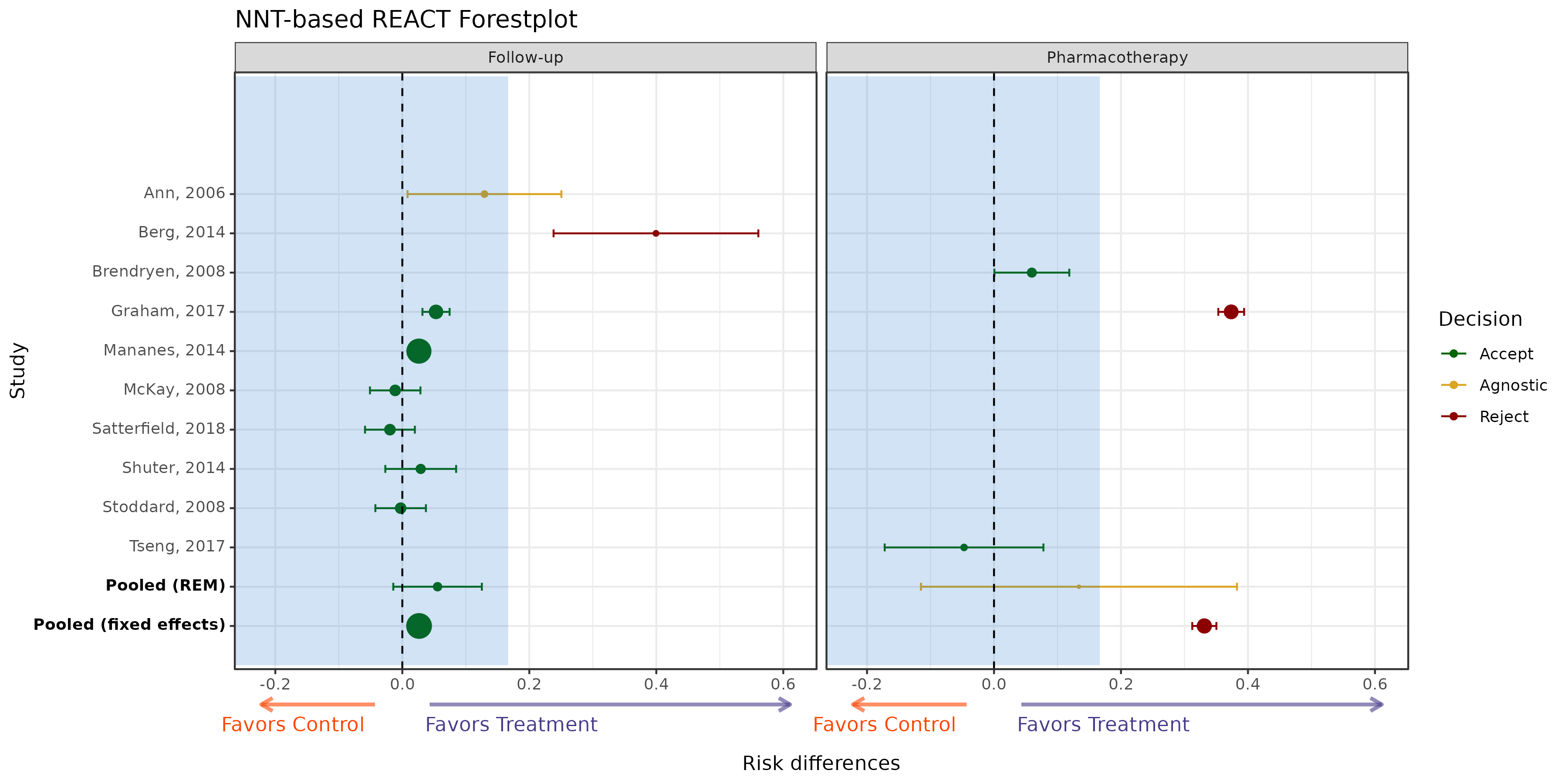}
    \caption{Forest plots for the meta-analyses described in Section \ref{sec:meta}, along with the decisions made by \ourapproach. The point estimate sizes are proportional to the inverse of the variances of each study. Our approach indicates that treatment is no better than control in the follow-up studies (that is, the null hypothesis is accepted), while it remains agnostic or is rejected in the pharmacotherapy studies, depending on the choice of the pooling method. The null region of equivalence was obtained using the Number Necessary to Treat (NNT).}
    \label{fig:nnt_forestplot}
\end{figure}

For the follow-up outcome studies, most of the intervals are contained in the region of equivalence, leading to the acceptance of the hypothesis that treatment is no better than control. In particular, the pooled results leads to the conclusion that the treatment is not worth pursuing.

On the other hand,
the pharmacotherapy studies are scarcer and do not agree with each other. 
This discrepancy results in a poor estimate for the REM pooled interval, leading to the decision of remaining agnostic, i.e., more studies need to be conducted to then reach an assertive conclusion. The fixed effects pooled interval leads to rejection, as the study that rejected the hypothesis is considerably more precise than the others. Therefore, the decision to remain undecided or reject the hypothesis depends on which pooling strategy the researcher finds more appealing for this case.

\section{Implications of the \ourapproach{} Method for Machine Learning and AI}
\label{sec:ml}

While this work has focused on formulations more adherent to Statistics, \ourapproach{} can also be extended to contexts akin to Machine Learning and AI.
By addressing   limitations of NHST, \ourapproach{} provides a principled way to evaluate models, handle multiple comparisons, and quantify uncertainty in automated decision-making systems. Here, we present a few of its advantages in this context.

\textbf{Sample Size Sensitivity in Large-Scale Machine Learning} In modern machine learning applications, datasets are often massive, which can lead traditional null hypothesis significance testing (NHST) to flag even negligible differences as statistically significant. This sensitivity to sample size is especially problematic when the goal is to assess whether a difference is practically relevant.
\ourapproach{} addresses this issue by testing \textbf{pragmatic} rather than \textbf{point-null} hypotheses. Instead of evaluating whether two models have exactly the same performance (e.g., $H_0: \mu_1 = \mu_2$), it tests whether the difference is within a region of equivalence (e.g., $H_0: |\mu_1 - \mu_2| \leq \Delta$), where $\Delta$ reflects a threshold of practical significance. This prevents trivial differences from leading to misleading conclusions and promotes more meaningful comparisons.
The benefits of \ourapproach{} are especially evident in supervised learning tasks, where $\mu_i$ could represent a performance metric such as AUC, F1 score, or accuracy. In A/B testing, it might represent the conversion rate of a strategy, while in reinforcement learning, it could be the expected value of a policy. By focusing on whether an observed effect is large enough to matter, rather than merely detectable, \ourapproach{} provides a robust framework for inference that remains consistent and interpretable even as sample sizes grow \citep{biecek2024rashomon, friedman2001elements}. 

\textbf{Three-Way Decision Framework in Model Selection and Ensembles} Model selection in machine learning is often based on significance testing, information criteria, or validation set performance. However, these methods do not account for practical equivalence. \ourapproach{}’s three-way decision framework (accept, reject, remain agnostic) offers a more nuanced alternative: when models are practically indistinguishable, selection can prioritize other criteria like interpretability or computational efficiency; when one model outperforms another by a meaningful margin, the better model is preferred \citep{biecek2024performance, domingos2012few}; and when results are inconclusive, further data collection or ensemble approaches may be warranted. In ensemble learning, \ourapproach{} can improve model diversity by selecting models that are practically different rather than merely statistically distinct, weighting ensemble components based on their unique predictive contributions, and organizing models hierarchically by clustering those that are practically equivalent. The logical coherence of \ourapproach{} ensures that these selection processes remain internally consistent, avoiding contradictions in multi-hypothesis settings.

\textbf{Error Rate Control and AutoML Reliability} Automated Machine Learning (AutoML) systems make numerous decisions regarding feature selection, hyperparameter tuning, and model architecture. NHST-based methods typically focus only on controlling Type I errors, often at the expense of Type II errors. \ourapproach{} improves AutoML reliability by controlling both Type I and Type II errors at level $\alpha$, preventing the inclusion of irrelevant features (Type I errors) and the exclusion of relevant ones (Type II errors), avoiding incorrect conclusions about model superiority when comparing architectures, and maintaining statistical power despite multiple testing corrections. By allowing an “agnostic” decision, \ourapproach{} provides AutoML systems with a principled way to defer decisions until more data is available, reducing overconfidence in uncertain cases \citep{automl2022survey}.

\textbf{Fairness and Bias in Machine Learning} Fairness evaluations in AI often rely on NHST to compare model performance across demographic groups. However, NHST-based fairness testing has notable limitations: large datasets make even small, practically irrelevant differences statistically significant; failing to reject a null hypothesis of “no difference” does not imply fairness; and multiple comparisons across demographic categories necessitate stringent corrections, reducing statistical power. \ourapproach{} addresses these challenges by defining fairness through practical equivalence ($H_0: |\mu_A - \mu_B| \leq \Delta$), where $\Delta$ represents an acceptable level of disparity; allowing three-way decisions, i.e., accepting fairness when evidence supports equivalence, rejecting fairness when disparities exceed meaningful thresholds, and remaining agnostic when evidence is inconclusive; handling multiple demographic comparisons with built-in FWER control; and ensuring logical consistency, preventing contradictory fairness conclusions \citep{mitchell2021ai, mehrabi2021survey}. By shifting the focus from statistical significance to practical significance, \ourapproach{} enables more interpretable and actionable fairness assessments in AI.

\textbf{Technical Challenges and Future Directions.} While \ourapproach{} offers clear benefits, its application to machine learning settings presents notable technical challenges. Many relevant performance metrics—such as AUC and F1 score—are aggregate, making the construction of valid confidence sets difficult. Existing approaches often rely on asymptotic approximations or resampling methods, which may not yield accurate coverage, especially in small-sample scenarios.  Moreover, in such cases with small samples, the resulting sets tend to be overly wide, leading frequently to agnostic outcomes.
 Another difficulty lies in specifying a meaningful threshold $\Delta$ for practical equivalence. For metrics like the F1 score, this choice is not straightforward and typically requires strong domain-specific knowledge. Furthermore, $\Delta$ may not scale linearly across different tasks or metrics. This motivates the development of more general dissimilarity functions that flexibly capture meaningful performance differences while being easier to calibrate.
These challenges point to important directions for future research: designing computationally efficient procedures for constructing confidence sets around interpretable, task-specific metrics, and developing principled methods for choosing $\Delta$ in practice.

\section{Final Remarks}
\label{sec:final}

We have shown the effectiveness of \ourapproach\ in overcoming numerous difficulties
with traditional NHSTs. A noteworthy aspect of \ourapproach\ is its ability to differentiate between ``evidence of absence'' and ``absence of evidence'' of a practical effect,  a crucial factor for conducting meta-analyses and comparing effects among different groups. Also, \ourapproach\ does not lead to an automatic rejection when the sample size is large.

We have argued that \ourapproach{}, building upon equivalence tests and three-way decision procedures, possesses numerous advantages over other approaches aimed at resolving NHST-related challenges.
Once the null hypothesis of practical interest is specified, \ourapproach\ only requires a confidence region to reach a decision, and thus
researchers from diverse domains can readily apply our approach.
Moreover, \ourapproach\ seamlessly integrates with confidence sets, which are widely regarded as more informative than hypothesis tests and p-values \citep{wasserstein2016asa}. If setting an equivalence region is feasible, this integration results in a framework that surpasses traditional NHST, enhancing the interpretability of statistical inferences.
Finally, \ourapproach\ can be used both within the frequentist and Bayesian frameworks.

Defining the threshold associated to the pragmatic region,  $\Delta$, is a delicate and domain-specific task. Poor choices of $\Delta$ may compromise the practical interpretation of results. To address this, we  recommend conducting sensitivity analyses over a range of $\Delta$ values to assess the robustness of conclusions. This anchors the choice of $\Delta$ in expert judgment and its observed consequences. In particular, it can be helpful to identify the smallest and largest $\Delta$ values for which the decision is non-agnostic. Moreover,
  the quality of \ourapproach's conclusions critically depends on the properties of the confidence or credibility set. If $C$ is constructed under violated assumptions (such as normality or independence), or if it is overly wide due to low power, the resulting decisions may be flawed or agnostic. In this sense, \ourapproach\ inherits the limitations of the underlying inferential machinery used to build $C$. 
  Notice however that although \ourapproach\ does not eliminate low power issues, it surfaces them transparently, avoiding misleading dichotomous conclusions.

Our examples and discussion focus on parametric models, where confidence intervals for means or differences are standard. Extending REACT to non-parametric or semi-parametric models presents additional challenges, particularly in constructing appropriate confidence/credibility regions and the associated computational costs. For the former, several strategies have been proposed to build regions with good properties \citep{genovese2005, robins2006, davies2009nonparametric, park2023}, especially in the Bayesian context, such as in \citet[][Section 6.3]{phadia2016} and \citet{lassance2024}. As for the latter, the computational cost can be circumvented when the delta method is available, but for more complex parameters using bootstrap or MCMC procedures might be required, which may be unreliable or too slow in certain contexts. In future work, we will adapt \ourapproach \ to these settings, accounting for such setbacks.



Adopting \ourapproach\ may contribute to increased confidence in scientific findings by making the decision-making process in hypothesis testing more transparent and by clearly presenting outcomes, including those cases where evidence is inconclusive.
This is particularly pertinent in the realms of meta-science and replication studies, which hold significant importance in the current scientific landscape.

To support the practical implementation of \ourapproach, we have developed an R package (available at \REACTlink{}) that simplifies its application for common models. This user-friendly package empowers researchers to efficiently implement our approach, fostering broader adoption and advancing scientific investigations. Embracing such tools and methodologies can contribute to more rigorous and reliable research practices, benefitting the scientific community as a whole.


%% file: sections/appendix.tex
\section{Additional Figures and Experiments}
\label{sec:add_figures}

\begin{figure}[!http]
    \centering
    \includegraphics[width = \textwidth]{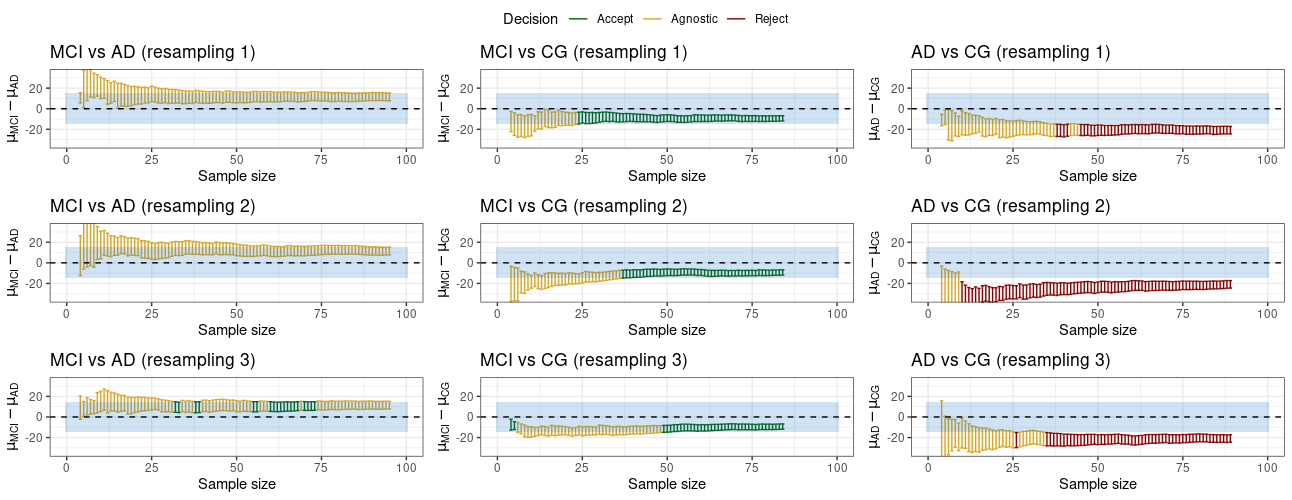}
    \caption{Confidence intervals for the average difference between groups as a function of the sample size for three different resamplings. All conclusions as the sample size increases are the same as the original sorting from \autoref{fig:pragmatic_camcog}, with slight differences in AD vs MCI for resampling 3.}
    \label{fig:camcog_resampling_plots}
\end{figure}

\begin{figure}[!http]
    \centering
    \includegraphics[width = .9\textwidth]{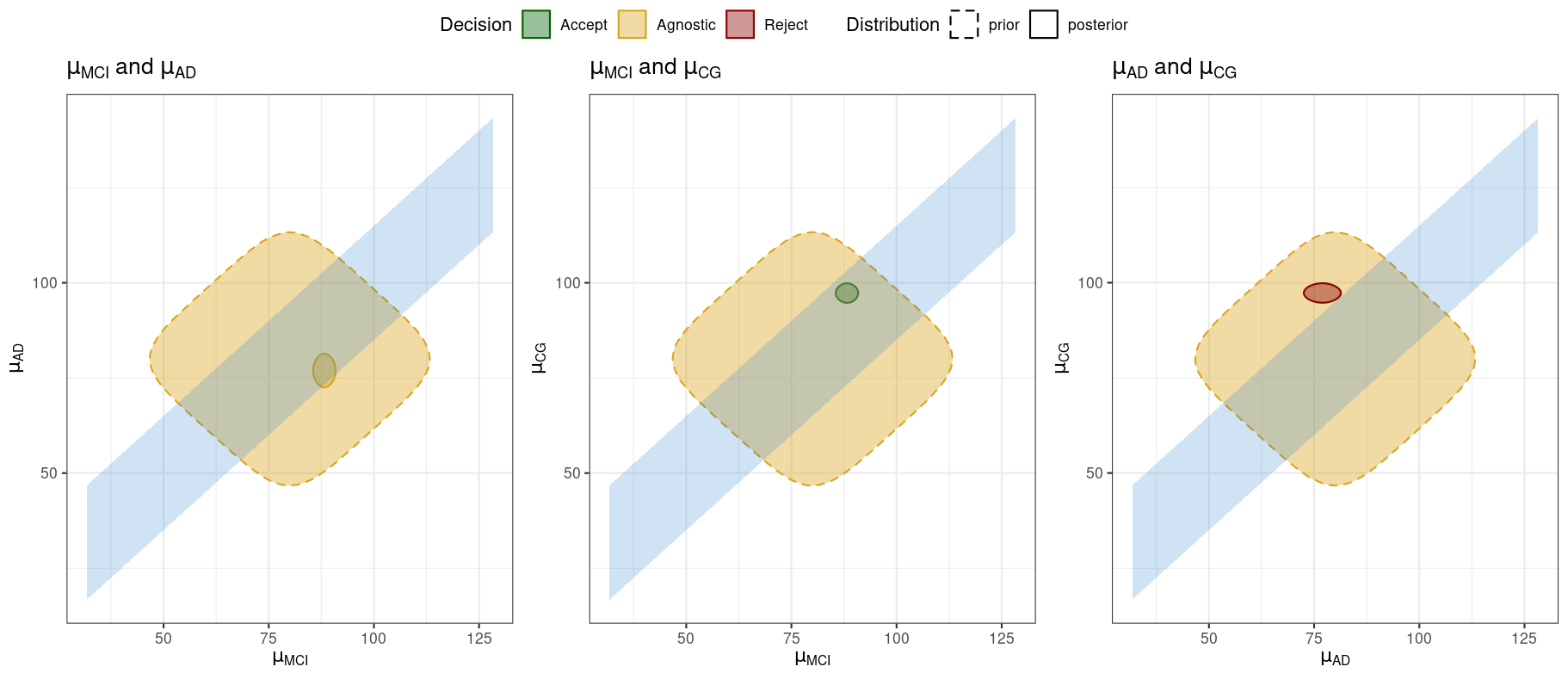}
    \caption{Bayesian pairwise group comparisons of \ourapproach\ for the CAMCOG application. The prior for each pair $(\mu_i, \sigma^2_i)$ is a Normal-inverse gamma with parameters $(80, 1, 3, 3)$. The credible regions are the HPD region of $(\mu_{AD}, \mu_{CG}, \mu_{MCI})$ projected on each pair $(\mu_i,\mu_j)$ (dashed border for prior, solid for posterior). The blue regions represent the null hypotheses $|\mu_i-\mu_j| \leq \Delta$ for the various groups. The conclusions obtained from the posterior are the same as those in \autoref{fig:camcog_react_ellipse}.}
    \label{fig:camcog_bayes}
\end{figure}

\begin{figure}[!http]
    \centering
    \includegraphics[width = .9\textwidth]{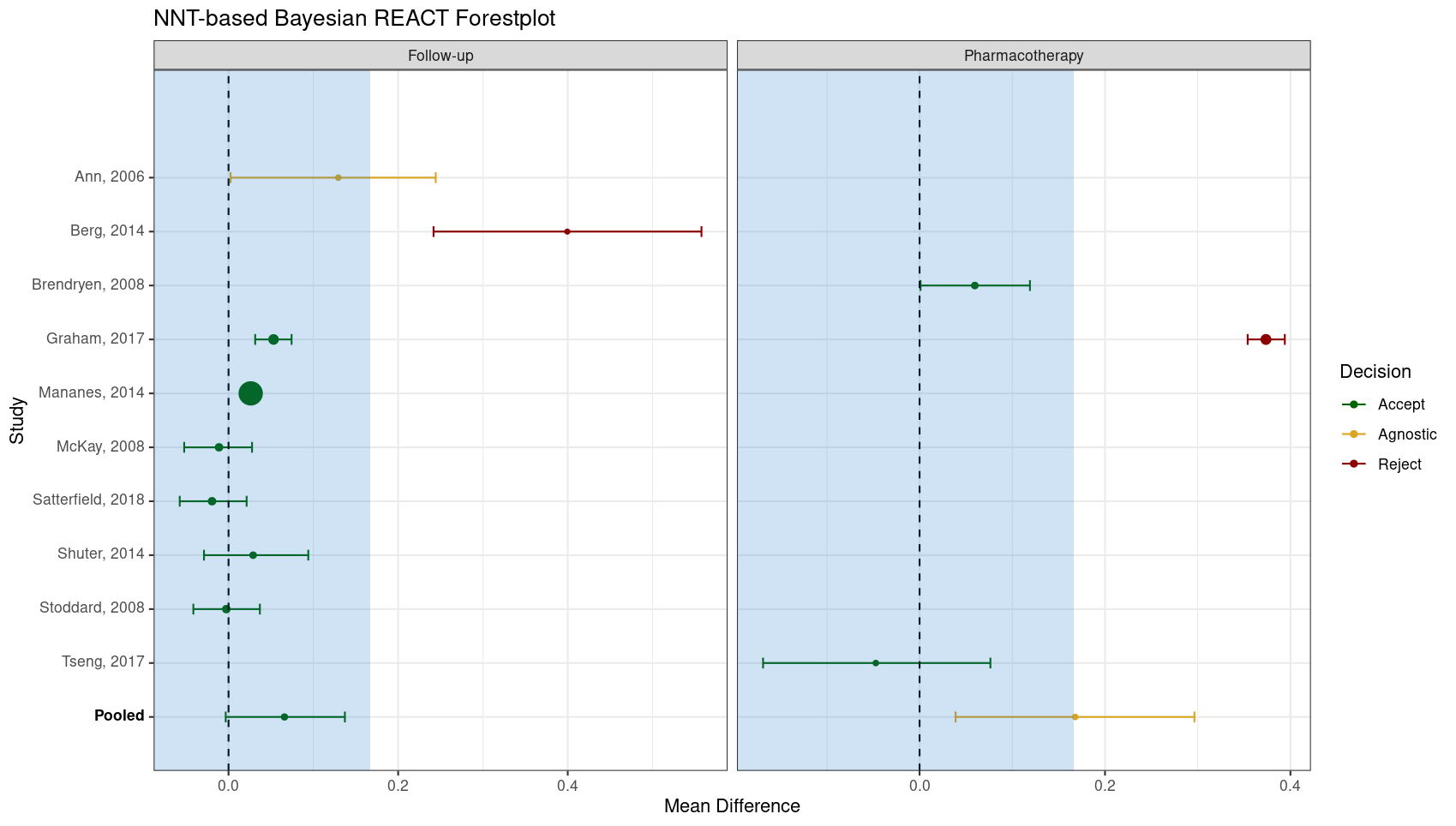}
    \caption{Bayesian REACT using the Jeffreys prior $Beta(1/2, 1/2)$ for the proportion of successes for the meta-analyses described in Section \ref{sec:meta}. The intervals represent the 95\% HPD region, while the point estimates are the posterior means and their respective sizes are the inverse of the posterior variances. The pooled model assumes a hierarchical structure, where the i-th study with the j-th intervention (control or treatment) has probability of success $\theta_{i,j}|\theta_j \sim \text{Logit-Normal}(\text{logit}(\theta_j), 0.1)$, while $\theta_j \sim Beta(1/2, 1/2)$. The conclusions are similar to those obtained in \autoref{fig:nnt_forestplot}.}
    \label{fig:forestplot_bayes}
\end{figure}

\section{Additional Theorems and Definitions}
\label{sec::addTheorems}

\begin{theorem}
\label{thm:nuisance}
\textbf{[Computation of \ourapproach\  in problems with nuisance parameters]}
Assume that the parameter space can be decomposed as $\Phi \times \Psi$, where $\phi \in \Phi$ denote parameters of interest and 
$\psi \in \Psi$ are nuisance parameters.
Let $C_\phi(\D) \subseteq \Phi$ be a region estimator of $\phi$ only. Consider the following   procedure to test 
    $H_0: \phi \in \Phi_0,$ where $\Phi_0 \subset \Phi$:
    \begin{equation*}
    \text{Decision}=
    \begin{cases}
      \text{\textbf{Accept} $H_0$} \ \text{if $C_\phi(\D) \subset \Phi_0$} \\
      \text{\textbf{Reject} $H_0$} \ \text{if $C_\phi(\D) \subset \Phi^c_0$} \\
      \text{\textbf{Remain Agnostic}} \ \text{otherwise}. \\
    \end{cases}
  \end{equation*}
  This is a proper \ourapproach\ procedure.
\end{theorem}

\begin{definition}
 \label{def:agnostic}
 A hypothesis test for a hypothesis $H \subset \Theta$ is a function  $\mathcal{R}_H: \mathcal{X} \longrightarrow \{0,1/2,1\}$, where $\mathcal{X}$ is the sample space, such that 0 represents the acceptance of $H$, 1 its rejection, and $1/2$ is the agnostic decision.
 \end{definition}

\begin{definition}
 \label{def:coherence}
 Let $\sigma(\Theta)$ be a subset of $\Theta$ with several hypotheses to be tested. For each $H \in \sigma(\Theta)$, let  $\mathcal{R}_H$ denote a test for $H$. The collection of hypothesis tests
 $(\mathcal{R}_H)_{H \in \sigma(\Theta)}$ is defined to be
 logically coherent if
 \begin{enumerate}
  \item (Propriety) $\mathcal{R}_\Theta \equiv 0$,
  \item (Monotonicity) If $H \subseteq H_*$,
  $\mathcal{R}_{H_*} \leq \mathcal{R}_H$,
  \item (Invertibility) For every $H$,
  $\mathcal{R}_H \equiv 1-\mathcal{R}_{H^c}$,
  \item (Intersection consonance) If
  $\mathcal{H}$ is a collection of 
  hypothesis such that
  $\mathcal{R}_H(D) = 0$, for every $H \in \mathcal{H}$, then
  $\mathcal{R}_{\cap_{H \in \mathcal{H}} H}(D) = 0$.
 \end{enumerate}
\end{definition}

\begin{definition}
 A region test $\mathcal{R}$, 
 has level $(\alpha,\beta)$ if
 $\sup_{H_0} \sup_{\theta \in H_0}
 \P_{\theta}(\mathcal{R}_{H_0} = 1) \leq \alpha$ and 
 $\sup_{H_0} \sup_{\theta \notin H_0}
 \P_{\theta}(\mathcal{R}_{H_0} = 0) \leq \beta$, that is
 the test controls the type I error by $\alpha$ and
 the type II error by $\beta$.
\end{definition}

\begin{definition}
 Let $\mathcal{R}$ be a region test.
 The decisiveness function of $\mathcal{R}$,
 $\beta_{\mathcal{R}}$, determines how
 frequently the test does not remain undecided, 
 for each value of $\theta$ and each
 hypothesis $H_0$, that is,
 \begin{align*}
  \beta_{\mathcal{R}}(\theta, H_0) = 
 \P_{\theta}\left(\mathcal{R}_{H_0} \neq \half\right).
 \end{align*}
\end{definition}

\begin{definition}
 A $(\alpha,\alpha)$-level region test, $\mathcal{R}$,
 is unbiased if 
 $\inf_{\theta, H_0} 
 \beta_{\mathcal{R}}(\theta, H_0) \geq \alpha$.
\end{definition}

\begin{definition}
 An interval region estimator, $\mathcal{C}$, is
 a region estimator such, 
 that there exist real-valued functions,
 $a(\mathcal{D})$ and $b(\mathcal{D})$ and, for every $\mathcal{D}$,
 $\mathcal{C}(\mathcal{D}) = 
 (a(\mathcal{D}), b(\mathcal{D}))$. 
\end{definition}

\begin{definition}
 A region test, $\mathcal{R}_{\mathcal{C}}$ is
 an interval region test if
 $\mathcal{C}$ is an interval region estimator.
\end{definition}

\begin{definition}
 Let $\mathcal{R}$ be an unbiased interval region test and
 $\mathcal{H}$ be a collection of hypotheses.
 $\mathcal{R}$ is uniformly most decisive 
 on $\mathcal{H}$ among
 unbiased $(\alpha,\alpha)$-level region tests 
 based on intervals if,
 for every unbiased, $(\alpha,\alpha)$-level,
 interval region test, $\mathcal{R}^*$,
 \begin{align*}
  \beta_{\mathcal{R}}(\theta,H_0) \geq
  \beta_{\mathcal{R}^*}(\theta,H_0),
  \text{ for every } \theta \text{ and } 
  H_0 \in \mathcal{H}.
 \end{align*}
\end{definition}

\begin{lemma}
 \label{lemma:unbiased}
 If $\mathcal{R}^*$ is an
 unbiased, $(\alpha,\alpha)$-level test based on
 the region $\mathcal{C}^*$, then
 $\phi^* = \I(\theta_{0} \notin \mathcal{C}^*)$ is
 an unbiased $\alpha$-level binary test.
\end{lemma}

\begin{theorem}
 \label{thm:best_c}
 Let $\mathcal{H} = \{\{\theta_0\}: \theta_0 \in \Theta\} \cup 
 \{(-\infty,\theta_0): \theta_0 \in \Theta\} \cup 
 \{(\theta_0,\infty): \theta_0 \in \Theta\}$, that is,
 the collection of all unilateral and bilateral hypotheses.
 For each $\theta$, assume $\phi_{\theta_0}$ is 
 a UMPU $\alpha$-level test for testing $H_0: \theta = \theta_0$. If
 $\mathcal{C}(\mathcal{D}) = \left\{\theta_0: 
 \phi_{\theta_0}(\mathcal{D}) = 0\right\}$ is 
 an interval region estimator, then
 $\mathcal{R}_{\mathcal{C}}$ is uniformly most decisive 
 on $\mathcal{H}$ among
 unbiased $(\alpha,\alpha)$-level interval region tests.
\end{theorem}

\begin{example}
 Let $\Phi$ be the cumulative density function of
 a standard normal, $X_1,\ldots,X_n$ be i.i.d. and
 $X_1 \sim N(\theta, \sigma^2_0)$, where
 $\theta$ is unknown and $\sigma^2_0$ is known.
 For each $H_0: \theta = \theta_0$,
 $\phi_{\theta_0} = \I\left(\frac
 {\sqrt{n}|\bar{X}-\theta_0|}{\sigma_0} \geq -\Phi(0.5\alpha) \right)$ is the UMPU $\alpha$-level test for $H_0$.
 Note that
 \begin{align*}
  \mathcal{C} &:= \{\theta_0: \phi_{\theta_0} = 0\} \\
 &= \left(\bar{X}
 +\frac{\Phi(0.5\alpha)\sigma_0}{\sqrt{n}}, 
 \bar{X}-\frac{\Phi(0.5\alpha)\sigma_0}{\sqrt{n}} \right)
 \end{align*}
 Hence, it follows from Theorem \ref{thm:best_c} that
 $\mathcal{R}_{\mathcal{C}}$ is uniformly most decisive 
 on $\mathcal{H}$ among
 unbiased $(\alpha,\alpha)$-level interval region tests.
\end{example}

%% file: sections/proofs.tex
\section{Proofs}
\label{sec:proofs}

\begin{lemma}
 \label{lemma:pval}
 Let $C(\D):=\left\{\theta \in \Theta: \text{p-val}_{\D}(\theta) > \alpha \right\}.$
 For every $H \subseteq \Theta$,
 $C(\D) \cap H = \emptyset$ if and only if
 $\max_{\theta \in H}\text{p-val}_{\D}(\theta) \leq \alpha$.
\end{lemma}

\begin{proof}
 If $\max_{\theta \in H}\text{p-val}_{\D}(\theta) \leq \alpha$, then
 for every $\theta \in H$, $\text{p-val}_{\D}(\theta) \leq \alpha$ and,
 by construction, $\theta \notin C(\D)$. That is,
 $C(\D) \cap H = \emptyset$.
 If $\max_{\theta \in H}\text{p-val}_{\D}(\theta) > \alpha$, then
 there exists $\theta \in H$, such that
 $\text{p-val}_{\D}(\theta) > \alpha$ and,
 by construction, $\theta \in C(\D)$. Therefore,
 $C(\D) \cap H \neq \emptyset$.
\end{proof}

\begin{proof}[Proof of Property \ref{proper:pvals}]
 First, observe that
 \begin{align*}
  \P_{\theta}(\theta \notin C(\mathcal{D}))
  &= \P_{\theta}(\text{p-val}_{\D}(\theta) \leq \alpha) \leq \alpha.
 \end{align*}
 Hence, $C(\mathcal{D})$ is a $1-\alpha$ confidence interval.
 The rest of the proof follows directly from 
 Lemma \ref{lemma:pval}.
\end{proof}

\begin{definition}
\label{def:region-test}
 Let $C(\mathcal{D})$ be 
 a region estimator for $\theta$, that is, a function $C: 
 \mathcal{X} \rightarrow \sigma(\Theta)$, where $\mathcal{X}$ is the sample space.
 The region test for $H_0$ based on $C$,
 $\mathcal{R}_{C,H_0}$, is
 \begin{align*}
  \mathcal{R}_{C,H_0}
  &= \begin{cases}
   0 & \text{, if } C \subseteq H_0 \\
   1 & \text{, if } C \subseteq H_0^c \\
   \half & \text{, otherwise.} 
  \end{cases}
 \end{align*}
\end{definition}

\begin{proof}[Proof of Property \ref{prop:projection-equivalence}]
    The proof is equivalent to showing that $\mathcal{R}_{C, H_0} = \mathcal{R}_{C_\Phi, \Phi_0}$. Also, when looking at the parameter space as a whole, we note that $H_0$ possesses the following property:
    \begin{equation}
    \label{eq:const}
        \left\{(\phi, \psi) \in H_0 \Longrightarrow \bigcup_{\psi \in \Psi}(\phi, \psi) \subset H_0\right\}, \quad \forall \phi \in \Phi.
    \end{equation}
    Now, let us consider the possible settings of $\mathcal{R}_{C, H_0}$. First, when the original test accepts $H_0$,
    \begin{equation}
        \label{eq:p8-ac}
        \mathcal{R}_{C, H_0} = 0 \stackrel{\text{def. }\ref{def:region-test}}{\Longleftrightarrow} C(\mathcal{D}) \subseteq H_0 \stackrel{\eqref{eq:const}}{\Longleftrightarrow} C_\Phi(\mathcal{D}) \times \Psi \subseteq H_0 = \Phi_0 \times \Psi \stackrel{\text{def. }\ref{def:region-test}}{\Longleftrightarrow} \mathcal{R}_{C_\Phi, \Phi_0} = 0.
    \end{equation}
    As for when the test rejects $H_0$,
    \begin{equation}
        \label{eq:p8-rej}
        \mathcal{R}_{C, H_0} = 1 \stackrel{\text{prop. }\ref{prop:symmetry}}{\Longleftrightarrow} \mathcal{R}_{C, H_0^c} = 0 \stackrel{\eqref{eq:p8-ac}}{\Longleftrightarrow} \mathcal{R}_{C_\Phi, \Phi_0^c} = 0 \stackrel{\text{prop. }\ref{prop:symmetry}}{\Longleftrightarrow} \mathcal{R}_{C_\Phi, \Phi_0} = 1.
    \end{equation}
    Lastly, since \eqref{eq:p8-ac} and \eqref{eq:p8-rej} are necessary and sufficient conditions for respectively accepting and rejecting $H_0$, it follows that $\mathcal{R}_{C, H_0} = 1/2 \Longleftrightarrow \mathcal{R}_{C_\Phi, \Phi_0} = 1/2$, thus concluding the proof.
    
\end{proof}

\begin{proof}[Proof of Property \ref{proper:control}] (Adapted from \citet{Coscrato2018}).
 Since $C(\mathcal{D})$ has confidence $1-\alpha$,
 $\P_{\theta}(\theta \notin C(\mathcal{D})) \leq \alpha$,
 for every $\theta \in \Theta$. Therefore,
 \begin{align*}
  \sup_{\theta_0 \in H_0}
  {\P_{\theta_0}(\mathcal{R}_{C,H_0}=1)} 
  &=  \sup_{\theta_0 \in H_0}
  {\P_{\theta_0}(C(\mathcal{D}) \subseteq H_0^c)}
  \leq \sup_{\theta_0 \in H_0}
  {\P_{\theta_0}(\theta_0 \notin C(\mathcal{D}))}
  \leq \alpha \\
  \sup_{\theta_1 \in H_0^c}
  {\P_{\theta_1}(\mathcal{R}_{C,H_0}=0)} 
  &=  \sup_{\theta_1 \in H_0^c}
  {\P_{\theta_1}(C(\mathcal{D}) \subseteq H_0)}
  \leq \sup_{\theta_1 \in H_0^c}
  {\P_{\theta_1}(\theta_1 \notin C(\mathcal{D})}
  \leq \alpha
 \end{align*}
\end{proof}

\begin{definition}
 The family-wise type I error, $FWER_I$,
 is the probability that some
 truly null hypothesis is incorrectly rejected.
 Similarly, the family-wise type II error, $FWER_{II}$,
 is the probability that some
  truly non-null hypothesis is incorrectly accepted. 
 That is,
 \begin{align*}
  FWER_I(\theta) 
  &:= \P_{\theta}(\cup_{H: \theta \in H} \mathcal{R}_{C,H} = 1) \\
  FWER_{II}(\theta) 
  &:= \P_{\theta}(\cup_{H: \theta \notin H} \mathcal{R}_{C,H} = 0).
 \end{align*}
\end{definition}

\begin{proof}[Proof of Property \ref{prop:fwer}]
 For every $H \in \sigma(\Theta)$,
 \begin{equation*}
  \begin{split}
   FWER_I(\theta)
   &= \P_{\theta}(\cup_{H: \theta \in H} \mathcal{R}_{C,H} = 1) \\
   &= \P_{\theta}(\cup_{H: \theta \in H} C \subseteq H^c) \\
   &= \P_{\theta}(\cup_{H: \theta \in H} C \cap H = \emptyset) \\
   &= \P_{\theta}(\theta \notin C) \leq \alpha.
  \end{split} 
  \hspace{2cm}
  \begin{split}
   FWER_{II}(\theta)
   &= \P_{\theta}(\cup_{H: \theta \notin H} \mathcal{R}_{C,H} = 0) \\
   &= \P_{\theta}(\cup_{H: \theta \notin H} C \subseteq H) \\
   &= \P_{\theta}(\theta \notin C) \leq \alpha. \\
   &
  \end{split}
\end{equation*}
\end{proof}

\begin{definition}
\label{def:convergeSet}
 A confidence set, $C(\mathcal{D})$,
 converges to the true $\theta$ if,
 for every $\theta_0 \in \Theta$ and
 $\epsilon$-ball around $\theta_0$,
 $B(\theta_0,\epsilon)$,
 \begin{align*}
  \limn \P_{\theta_0}(C(\mathcal{D}) 
  \subseteq B(\theta_0, \epsilon)) = 1.
 \end{align*}
\end{definition}

\begin{definition}
 \label{def:ic_converge}
 For $A \subseteq \Theta$, let
$\abxring{A}$ denote the interior of $A$. 
 A hypothesis test for $H_0$, 
 $\mathcal{R}_{H_0}$, is consistent if,
 for every $\theta \in \abxring{H_0}$,
 $\limn \P_{\theta}(\mathcal{R}_{H_0} = 0) = 1$ and,
 for every $\theta \in \abxring{H_0^c}$,
 $\limn \P_{\theta}(\mathcal{R}_{H_0} = 1) = 1$.
\end{definition}

\begin{proof}[Proof of Property 
 \ref{proper:consistency}]
 Let $\theta \in \abxring{H_0}$.
 There exists $\epsilon > 0$ such that
 $B(\theta, \epsilon) \subseteq H_0$.
 Hence,
 \begin{align*} 
  \limn \P_{\theta}(\mathcal{R}_{H_0,C} = 0) 
  &= \limn \P_{\theta}(C(\mathcal{D}) \subseteq H_0) \\
  &\leq \limn \P_{\theta}(C(\mathcal{D}) \subseteq B(\theta,\epsilon)) 
  & B(\theta,\epsilon) \subseteq H_0 \\
  &= 1 & \text{Definition \ref{def:ic_converge}}
 \end{align*}
 Similarly, let $\theta \in \abxring{H_0^c}$.
 There exists $\epsilon > 0$ such that
 $B(\theta, \epsilon) \subseteq H_0^c$. Hence,
 \begin{align*} 
  \limn \P_{\theta}(\mathcal{R}_{H_0,C} = 1) 
  &= \limn \P_{\theta}(C(\mathcal{D}) \subseteq H_0^c) \\
  &\leq \limn \P_{\theta}(C(\mathcal{D}) \subseteq B(\theta,\epsilon)) 
  & B(\theta,\epsilon) \subseteq H_0^c \\
  &= 1 & \text{Definition \ref{def:ic_converge}}
 \end{align*}
\end{proof}

\begin{proof}[Proof of Property \ref{proper:coherence}] (Adapted from \citet{Esteves2016}).
 Let $(\mathcal{R}_{H})_{H \in \sigma(\Theta)}$ be 
 a collection of tests based on
 confidence set $C$.
 \begin{enumerate}
  \item Since $C(\mathcal{D}) \subseteq \Theta$,
  $\mathcal{R}_{\Theta} \equiv 0$,
  \item Let $H \subseteq H_*$. 
  If $\mathcal{R}_H(D) = 0$, then $C(D) \subset H$.
  Hence, $C(D) \subset H_*$, that is,
  $\mathcal{R}_{H_*}(D) = 0$. Also,
  if $\mathcal{R}_{H_*}(D) = 1$, then
  $C(D) \subseteq H_*^c$. Hence,
  $C(D) \subseteq H^c$ and
  $\mathcal{R}_{H}(D) = 1$. Conclude that
  $\mathcal{R}_{H_*} \leq \mathcal{R}_H$.
  \item It is sufficient to prove that,
  for every $H \in \sigma(\Theta)$,
  $\mathcal{R}_H(D) = 0$ if and only if 
  $\mathcal{R}_{H^c}(D) = 1$.
  The proof follows from the fact that
  $\mathcal{R}_H(D) = 0$ when 
  $C(D) \subseteq H$ and
  $\mathcal{R}_{H^c}(D) = 1$ when
  $C(D) \subseteq (H^c)^c$, that is,
  $C(D) \subseteq H$.
  \item If $\mathcal{R}_H(D) = 0$, 
  for every $H \in \mathcal{H}$, then
  $C(D) \subseteq H$, for every $H \in \mathcal{H}$.
  Hence, $C(D) \subseteq \cap_{H \in \mathcal{H}} H$.
  That is, $\mathcal{R}_{ \cap_{H \in \mathcal{H}} H }(D) = 0$.
 \end{enumerate}
\end{proof}

\begin{proof}[Proof of Property \ref{prop:symmetry}]
 Property \ref{prop:symmetry} is
 a consequence of invertibility
 in Definition \ref{def:coherence}
 Hence, this property is 
 a corollary of 
 Property \ref{proper:coherence}.
\end{proof}

\begin{proof}[Proof of Theorem \ref{thm:nuisance}]
    Let   $C(\D):=C_\phi(\D) \times \Psi$. Notice that
    $$C_\phi(\D) \subset \Phi_0  \iff C(\D) \subset \Phi_0 \times \Psi,$$
and therefore the procedure stated on the theorem is equivalent to the following \ourapproach\ procedure:
\begin{equation*}
    \text{Decision}=
    \begin{cases}
      \text{\textbf{Accept} $H_0$} \ \text{if $C(\D) \subset \Phi_0 \times \Psi$} \\
      \text{\textbf{Reject} $H_0$} \ \text{if $C(\D) \subset \Phi^c_0 \times \Psi$} \\
      \text{\textbf{Remain Agnostic}} \ \text{otherwise}. \\
    \end{cases}
  \end{equation*}
The conclusion follows.
\end{proof}

\begin{proof}[Proof of Lemma \ref{lemma:unbiased}]
 Since $\mathcal{R}^*$ is a region-based test,
 $\{\mathcal{R}^* \text{ rejects } \{\theta_0\}\}$
 is the same as $\{\phi^* = 1\}$. Hence,
 since $\mathcal{R}^*$ is a $(\alpha,\alpha)$-level test
 \begin{align*}
  \P_{\theta_0}(\phi^* = 1) &=
  \P_{\theta_0}(\mathcal{R}^* \text{ rejects } \{\theta_0\})
  \leq \alpha.
 \end{align*}
 That is, $\phi^*$ is
 a $\alpha$-level binary test.
 Also, since $\mathcal{R}^*$ is unbiased,
 \begin{align*}
  \P_{\theta}(\phi^* = 0) &=
  \P_{\theta_0}(\mathcal{R}^* \text{ remains undecided about } \{\theta_0\}) \leq 1-\alpha
 \end{align*}
 Hence, $\phi^*$ is unbiased.
\end{proof}

\begin{proof}[Proof of Theorem \ref{thm:best_c}]
 Let $\mathcal{R}^*$ be an
 unbiased, $(\alpha,\alpha)$-level
 based on the interval region $\mathcal{C}^*$.
 For every $H_0 \in H$, define
 $\phi^* = \I(\theta_{0} \notin \mathcal{C}^*)$.
 \begin{align*}
  \beta_{\mathcal{R}}(\theta,H_0)
  &= \P_{\theta}(\theta_0 \notin \mathcal{C}) \\
  &= \P_{\theta}(\phi_{\theta_0} = 1) \\
  & \geq \P_{\theta}(\phi^* = 1) 
  & \text{Lemma \ref{lemma:unbiased} and 
  $\phi_{\theta_0}$ is UMPU} \\
  &= \P_{\theta}(\theta_0 \notin \mathcal{C}^*) \\
  &= \beta_{\mathcal{R}^*}(\theta,H_0).
 \end{align*}
\end{proof}

\section{Additional Results}

\begin{definition}
 The Bayesian family-wise
 false conclusion error of
 $(\mathcal{R}_H)_{H \in \sigma(\Theta)}$, 
 $\gamma$, is:
 \begin{align*}
  \gamma &=
  \P(\exists H \in \sigma(\Theta):
  (\theta \in H \text{ and } \mathcal{R}_H = 1) 
  \text{ or } 
  (\theta \notin H \text{ and } \mathcal{R}_H = 0))
 \end{align*}
\end{definition}

\begin{theorem}
 \label{thm:family-wise-bayes}
 If $R$ is a credibility region for $\theta$
 with credibility $1-\alpha$ and
 $(\mathcal{R}_H)_{H \in \sigma(\Theta)}$ is
 based on $R$, then
 $\gamma \leq \alpha$.
\end{theorem}

\begin{proof}[Proof of Theorem \ref{thm:family-wise-bayes}]
 \begin{align*}
  \gamma &=
  \P(\exists H \in \sigma(\Theta):
  (\theta \in H \text{ and } \mathcal{R}_H = 1) 
  \text{ or } 
  (\theta \notin H \text{ and } \mathcal{R}_H = 0)) \\
  &= \P(\exists H \in \sigma(\Theta):
  (\theta \in H \text{ and } R \subseteq H^c)
  \text{ or } 
  (\theta \notin H \text{ and } R \subseteq H)) \\
  &\leq \P(\theta \notin R) = \alpha
 \end{align*}
\end{proof}

\begin{theorem}
    \textbf{[Computation of \fourapproach\ using p-values in problems with nuisance parameters]}
\label{thm:pvals_nuisance}
Assume that the parameter space can be decomposed as $\Phi \times \Psi$, where $\phi \in \Phi$ denote parameters of interest and 
$\psi \in \Psi$ are nuisance parameters.
    Let $\text{p-val}_{\D}(\phi_0)$ be a  p-value for the hypothesis $H_0:\phi=\phi_0$. Consider the following   procedure to test 
    $H_0: \phi \in \Phi_0,$ where $\Phi_0 \subset \Phi$:
    \begin{equation*}
    \text{Decision}=
    \begin{cases}
      \text{\textbf{Accept} $H_0$} \ \text{if $\sup_{\phi \in \Phi^c_0}\text{p-val}_{\D}(\phi) \leq \alpha$} \\
      \text{\textbf{Reject} $H_0$} \ \text{if $\sup_{\phi \in \Phi_0} \text{p-val}_{\D}(\phi) \leq \alpha$} \\
      \text{\textbf{Remain Agnostic}} \ \text{otherwise} \\
    \end{cases}
  \end{equation*}
  This procedure is equivalent to the following \fourapproach\ procedure:
\begin{equation*}
    \text{Decision}=
    \begin{cases}
      \text{\textbf{Accept} $H_0$} \ \text{if $C(\D) \subset \Phi_0 \times \Psi$} \\
      \text{\textbf{Reject} $H_0$} \ \text{if $C(\D) \subset \Phi^c_0 \times \Psi$} \\
      \text{\textbf{Remain Agnostic}} \ \text{otherwise}. \\
    \end{cases}
  \end{equation*}
   where $C(\D):=C_\phi(\D) \times \Psi$
   and 
     $C_\phi(\D):=\left\{\phi \in \Phi: \text{p-val}_{\D}(\phi)  > \alpha \right\}.$
   Moreover, it can be more easily written as
    \begin{equation*}
    \text{Decision}=
    \begin{cases}
      \text{\textbf{Accept} $H_0$} \ \text{if $C_\phi(\D) \subset \Phi_0$} \\
      \text{\textbf{Reject} $H_0$} \ \text{if $C_\phi(\D) \subset \Phi^c_0$} \\
      \text{\textbf{Remain Agnostic}} \ \text{otherwise}. \\
    \end{cases}
  \end{equation*}
 Also, $C_\phi(\D)$ is a (1-$\alpha$)-level confidence set for $\phi$, and 
 $C(\D)$ is a (1-$\alpha$)-level confidence set for $(\phi,\psi)$.
\end{theorem}

\begin{proof}
    Notice that 
    \begin{align*}
        \sup_{\phi \in \Phi_0^c} \text{p-val}_{\D}(\phi) \leq \alpha &\iff \text{For every $\phi \in \Phi_0^c$,  } \text{p-val}_{\D}(\phi) \leq \alpha \\ 
        &\iff  \text{For every $\phi \in \Phi_0^c$ and $\psi \in \Psi$,  } (\phi,\psi) \notin C(\D)\\ 
        &\iff C(\D) \subset \Phi_0 \times \Psi \\ 
    \end{align*}
    Thus, the procedure accepts $H_0$ if, and only if, $C(\D) \subset \Phi_0 \times \Psi$. Similarly, the procedure rejects $H_0$ if, and only if, $C(\D) \subset \Phi^c_0 \times \Psi$. It follows that this procedure is a \ourapproach-type procedure. Now, if the p-values are valid, for every $(\phi,\psi) \in \Phi \times \Psi$,
    $$\P_{(\phi,\psi)}\left((\phi,\psi) \in C(\D)\right)=\P_{(\phi,\psi)} \left(\text{p-val}_{\D}(\phi)>\alpha \right)=1-\alpha,$$
which concludes the proof.
\end{proof}

\begin{definition} [\citet{Pereira1999}]
    The e-value for a hypothesis $H_0:\theta=\theta_0$ is the posterior probability
    $$\text{e-val}_{\D}(\theta_0)=1-\P\left(\theta \in T_{\D} |\D\right),$$
    where 
    $$T_{\D} = \{\theta: f(\theta|\D) \geq f(\theta_0|\D)\} $$
\end{definition}

\begin{theorem} \textbf{[Computation of \bourapproach\ using e-values]}
\label{thm:evals}
    Let $\text{e-val}_{\D}(\theta_0)$ be an  e-value \citep{Pereira1999} for the hypothesis $H_0:\theta \in \Theta_0$. Then  the following  is a \bourapproach \ procedure:
    \begin{equation*}
    \text{Decision}=
    \begin{cases}
      \text{\textbf{Accept} $H_0$} \ \text{if $\max_{\theta \in \Theta_0^c} \text{e-val}_{\D}(\theta) \leq \alpha$} \\
      \text{\textbf{Reject} $H_0$} \ \text{if $\max_{\theta \in \Theta_0}\text{e-val}_{\D}(\theta) \leq \alpha$} \\
      \text{\textbf{Remain Agnostic}} \ \text{otherwise} \\
    \end{cases}
  \end{equation*}
    The  (1-$\alpha$)-level Bayes
    set that corresponds to this procedure is the ($1-\alpha$)-level Highest Posterior Density (HPD) region for $\theta$:
  $$C(\D):=\left\{\theta \in \Theta: \pi(\theta|\D)  > C \right\},$$
  where $C$ is such that 
  $$\P\left(\theta \in C(\D)| \D \right)=1-\alpha.$$
  This procedure is equivalent to the GFBST \citep{Stern2017} and, if $C$ is an interval, to ROPE \citep{kruschke2018rejecting}.
\end{theorem}

\begin{proof}
    Can be found in \citet[Example 8]{Esteves2016}.
\end{proof}